\documentclass[10pt, conference, letterpaper]{IEEEtran}
\ifCLASSINFOpdf
 \usepackage[pdftex]{graphicx}
 \usepackage{subcaption}
\else
\fi
\usepackage{color}
\usepackage{amsmath}
\usepackage{amsthm}
\usepackage{mathtools}
\usepackage{amssymb}
\usepackage{bbm}
\usepackage{cite}
\usepackage{optidef}
\usepackage{steinmetz}
\usepackage[linesnumbered,lined,ruled,commentsnumbered]{algorithm2e}
\SetAlFnt{\small\sffamily}

\usepackage[only,llbracket,rrbracket,llparenthesis,rrparenthesis]{stmaryrd}
\usepackage{enumerate}

\theoremstyle{definition}

\newtheoremstyle{mytheorem}
  {3pt}
  {3pt}
  {\itshape}
  {}
  {\itshape\bfseries}
  {.}
  {.5em}
  {\thmname{#1}\thmnumber{ #2} \thmnote{ {\the\thm@notefont(#3)}}}

\makeatother
\theoremstyle{mytheorem}

\newtheorem{theorem}{Theorem}

\newtheorem{lemma}[theorem]{Lemma}
\newtheorem{proposition}[theorem]{Proposition}
\theoremstyle{remark}
\newtheorem*{remark}{Remark}

\makeatletter
\renewcommand*\env@matrix[1][*\c@MaxMatrixCols c]{%
  \hskip -\arraycolsep
  \let\@ifnextchar\new@ifnextchar
  \array{#1}}
\makeatother

\makeatletter
\patchcmd{\@makecaption}
  {\scshape}
  {}
  {}
  {}
\makeatother

\newcommand\Tstrut{\rule{0pt}{2.6ex}}         
\newcommand\Bstrut{\rule[-0.9ex]{0pt}{0pt}}   

\newcommand\blfootnote[1]{%
  \begingroup
  \renewcommand\thefootnote{}\footnote{#1}%
  \addtocounter{footnote}{-1}%
  \endgroup
}

\newcommand{\norm}[1]{\left\lVert#1\right\rVert}

\DeclareMathOperator*{\argmin}{arg\,min} 

\begin{document}
\title{Beam Discovery Using Linear Block Codes for Millimeter Wave Communication Networks}

\author{\IEEEauthorblockN{Yahia Shabara, C. Emre Koksal and Eylem Ekici}
\IEEEauthorblockA{Department of Electrical and Computer Engineering\\
The Ohio State University, Columbus, Ohio 43210\\
Email: \{shabara.1, koksal.2, ekici.2\}@osu.edu}
}
\maketitle



\IEEEpeerreviewmaketitle

\begin{abstract}
The surge in mobile broadband data demands is expected to surpass the available spectrum capacity below $6$ GHz. This expectation has prompted the exploration of millimeter wave (mm-wave) frequency bands as a candidate technology for next generation wireless networks. However, numerous challenges to deploying mm-wave communication systems, including channel estimation, need to be met before practical deployments are possible.
This work addresses the mm-wave channel estimation problem and treats it as a beam discovery problem in which locating beams with strong path reflectors is analogous to locating errors in linear block codes.
We show that a significantly small number of measurements (compared to the original dimensions of the channel matrix) is sufficient to reliably estimate the channel. We also show that this can be achieved using a simple and energy-efficient transceiver architecture.
\blfootnote{Eylem Ekici is supported in part by NSF grants CNS-1421576 and CNS-1731698.
C. Emre Koksal is supported in part by NSF grants CNS-1618566, CNS-1514260.}
\end{abstract}

\section{Introduction}
\label{Intro}
We investigate the problem of channel estimation in millimeter wave (mm-wave) wireless communication networks.
Mm-wave refers to the wavelength of electromagnetic signals at 30-300 GHz frequency bands.
At these high frequencies, channel measurement campaigns revealed that wireless communication channels exhibit very limited number of scattering clusters in the angular domain \cite{akdeniz2014millimeter, rangan2014millimeter, anderson2004building}. A \textit{cluster} refers to a propagation path or continuum of paths that span a small interval of transmit Angles of Departure (AoD) and receive Angles of Arrival (AoA).
Moreover, signal attenuation is very significant at mm-wave frequencies.
This motivates the use of large antenna arrays at the transmitter (TX) and receiver (RX) to provide high antenna gains that compensate for high path losses \cite{rappaport2013millimeter}.
Nevertheless, due to the high power consumption of mixed signal components, e.g., Analog to Digital Converters (ADCs) \cite{DSPforMmWave2016}, conventional digital transceiver architectures that employ a complete RF chain per antenna is not practical.
%
Hence, alternate architectures have been proposed for mm-wave radios with the objective of maintaining a close performance to channel capacity. Among the proposed solutions are the use of i) hybrid analog/digital beamforming \cite{alkhateeb2014channel,mendez2015channel,han2015large} and ii) fully digital beamforming with low resolution ADCs \cite{ImpactOfResOnPerf2015,highSNRcapacitySingleBitADC2014,AdaptiveOneBit_Rusu2015}.

For all proposed solutions, \textit{channel estimation} remains one of the most critical determinants of performance in communication.
Due to the large number of antennas at TX and RX,
estimation of the full channel gain matrix
may require a large number of measurements, proportional to the product of the number of transmit and receive antennas. This imposes a great burden on the estimation process. To address this issue, various methods have been used, the most prevalent among them, is compressed sensing \cite{Bajwa2010_compressedChannelSensing, ChEstSunAdaptiveCS2017, alkhateeb2014channel,ChEstScniter14,AdaptiveOneBit_Rusu2015}, which leverages channel sparsity.
Performance of compressed sensing based approaches is heavily dependent on the design of system (sensing) matrices. For instance, while random sensing matrices are known to perform well, in practice, sensing matrices involve the design of transmit and receive beamforming vectors and the choice of dictionary matrices\footnote{A dictionary matrix is used to express the channel in a sparse form.}.
Hence, purely random matrices have not been used in practice \cite{CS4Wireless_TipsAndTricks}. On the other hand, no design that involves deterministic sensing matrices has been considered for sparse channel estimation.

Despite the efforts, we do not have a full understanding of the dependence of channel estimation performance on the channel parameters and number of measurements.
In an effort to understand this relationship, the study in \cite{Alkhateeb_2015_HowManyMeasurements} proposed a multi-user mm-wave downlink framework based on compressed sensing in which the authors evaluate the achievable rate performance against the number of measurements.


In this work, we follow a different approach. We propose a systematic method in which we use sequences of error correction codes chosen in a way to control the channel estimation performance.
To demonstrate our approach, consider the following simple example.
Let a point to point communication channel be such that, there exists 3 possible receive AoA directions, only one of which may have a strong path to TX. We need to obtain the correct AoA at RX, if it exists.
Instead of exhaustively searching all 3 possible AoA directions, we alternatively measure signals from combined directions. For instance, by combining directions 1\&2 in one measurement and 2\&3 in the next measurement, we can find the AoA in just two measurements.
Specifically, four different scenarios might occur, namely, i) only the $1^{\text{st}}$, or ii) only the $2^{\text{nd}}$ measurement contains a strong path, iii) both $1^{\text{st}}$ and $2^{\text{nd}}$ measurements contain a strong path, and finally, iv) neither measurement reveals a strong path.
Interpretation of those cases is: AoA is in  i) direction 1, ii) direction 3, iii) direction 2, and iv) none exists.
Therefore, only 2 measurements are sufficient for beam detection instead of 3 that are needed for exhaustive search.


We will generalize this idea to develop a systematic method for beam detection, inspired by linear block coding.
Specifically, we show that linear block error correcting codes (LBC) possess favorable properties that fit in with the desirable behavior of sparse channel estimation. As a result, we are able to \textbf{i) provide rigorous criteria for solving the channel estimation problem}, \textbf{ii) significantly decrease the number of required measurements}, and \textbf{iii) utilize a fairly simple and energy-efficient transceiver architecture.}
We design the system using LBCs that leverage the fact that transmission errors are typically sparse in transmitted data streams, and hence, only a few number of erroneous bits need to be corrected per transmitted codeword.
Similarly, mm-wave channels are also sparse, i.e., only a small number of AoAs/AoDs carry strong signals.
LBCs can correct sparse transmission errors by identifying their location in a transmitted sequence (followed by flipping them).
We are inspired by LBC's ability to locate erroneous bits and exploit it to identify the AoAs/AoDs that carry strong signals (and their path gains) among all possible AoA/AoD values.
To this end, we exploit hard decision decoding of LBCs, in which the receiver obtains an \textit{error syndrome} that maps to one of the correctable error patterns.
An obtained error pattern determines the positions where errors have occurred.
Likewise, for channel estimation, the receiver will be designed to do a sequence of measurements that would result in a \textit{channel syndrome}. The resultant channel syndrome shall identify the positions (and values) of non-zero angular channel components.

Contributions of this work can be summarized as follows:
\begin{itemize}
\item	We set an analogy between beam discovery and channel coding to utilize low-complexity decoding techniques for efficient beam discovery.
\item	We provide rigorous criteria for setting the number of channel measurements based on the size of the channel and its sparsity level.
\item   We show that the number of measurements required for beam discovery is linked to the rate of a used linear block code. Hence, maximizing the rate of the underlying code is equivalent to minimizing the number of measurements.
\item	We develop a simple receiver architecture that enables us to measure signals arriving from multiple directions.
\end{itemize}

\textbf{Related Work:}
The main objective of mm-wave channel estimation is to find a mechanism that can reliably estimate the channel using as few measurements as possible.
For instance, in \cite{alkhateeb2014channel}, a compressed sensing based algorithm to estimate single-path channels is proposed and an upper bound on its estimation error is derived. Further, the authors propose a multipath channel estimation algorithm based on that of single-path channels. The proposed algorithms in \cite{alkhateeb2014channel} use an adaptive approach with a hierarchical codebook\footnote{A codebook refers to the set of all possible beamforming vectors.}
of increasing resolution. Similarly, the work in \cite{ChEstSunAdaptiveCS2017} proposes an adaptive compressive sensing channel estimation algorithm that accounts for off-the-grid AoAs and AoDs by using continuous basis pursuit \cite{CBP_2011} dictionaries.
Such adaptive algorithms divide the estimation process into stages and demand frequent feed back to the TX after each stage.
Hence, while the number of required measurements are shown to decrease, these methods may add a considerable overhead.

Other works like \cite{AgileMmWave_arxiv17,AgileMmWave_HotNets16} and \cite{RACE_2017} have proposed channel estimation algorithms using overlapped beam patterns.
For instance, the algorithm in \cite{RACE_2017} can estimate multipath channel components by sequentially estimating each path gain using an algorithm designed to estimate single-path channels followed by recursively removing the estimated paths' effect from subsequent measurements. Similar to \cite{alkhateeb2014channel,ChEstSunAdaptiveCS2017} adaptive beams with increasing resolution that require feedback to TX are used to refine the AoA/AoD estimates.
On the other hand, the beam alignment algorithms proposed in \cite{AgileMmWave_arxiv17,AgileMmWave_HotNets16} assume a multipath mm-wave channel. These algorithms, with a high probability, can find the best beam alignment in a logarithmic number of measurements (with respect to the total number of available AoA directions). Nonetheless, despite the possible existence of multiple paths, those algorithms are designed to find one path to TX.

Exploiting the results of previous beam alignment operations could be used to reduce the overhead of subsequent alignments.
For instance, assuming that successive beam alignments are statistically correlated, the authors in \cite{hashemi2017efficient}
use this contextual information to improve beamforming delay via Multi-Armed Bandit based models.

Most research efforts in the field of mm-wave channel estimation use the magnitude and phase information of the acquired channel measurements. Nevertheless, if a carrier frequency offset (CFO) error occurs in the transceiver hardware, the phase information might be unreliable. Hence, the work in \cite{AgileMmWave_arxiv17,AgileMmWave_HotNets16,NonCoherentPathTrack_17} tackle this problem by ignoring the phase information.
Similar to \cite{AgileMmWave_arxiv17,AgileMmWave_HotNets16}, the solution in \cite{NonCoherentPathTrack_17} can only obtain one (dominant) path between TX and RX using a compressed sensing based technique.
The CFO problem is tackled in \cite{myers2017compressive} by considering it as a variable to be estimated.

While the power consumption problem of mmwave systems is commonly alleviated using analog or hybrid beamforming transceivers, an alternative solution is to use low resolution ADCs in fully digital architectures.
Owing to the fact that low resolution ADCs operate at much lower power than their high resolution counterparts, the work in \cite{ImpactOfResOnPerf2015,highSNRcapacitySingleBitADC2014,AdaptiveOneBit_Rusu2015,alkhateeb2014mimo}
employ low resolution (single-bit) ADCs in digital transceivers.
The work in \cite{ChEstScniter14,barati2015directional} study the channel estimation problem using such architectures. 
Moreover, other solutions include integrated mm-wave and sub-6 GHz systems \cite{hashemi2017energy} to provide reliable and energy efficient communication systems.

\subsubsection*{Notations} A vector and a matrix are denoted by $\boldsymbol{x}$ and $\boldsymbol{X}$, respectively, while $x$ denotes a scalar or a complex number depending on the context. The transpose, conjugate transpose and frobenius norm of $\boldsymbol{X}$ are given by $\boldsymbol{X}^T$, $\boldsymbol{X}^H$ and $\left\Vert \boldsymbol{X} \right\Vert_F$, respectively.
The sets of real and complex numbers are $\mathbb{R}$ and $\mathbb{C}$.
The $k {\times} k$ identity matrix is $\boldsymbol{I_k}$.
A set is denoted by $\mathcal{X}$, while $|\mathcal{X}|$ is its cardinality. Finally, $\mathbbm{1}()$ is the indicator function.

\section{Motivating Example}
\label{illustrative}
To elaborate, we present the following example:
consider a point to point communication link between a TX with single antenna ($n_t = 1$) and RX with $n_r = 15$ antennas. Therefore, the vector of channel gains\footnote{Let all the channels have one single significant tap.},
$\boldsymbol{q}$, is a $15 {\times} 1$ vector,
and its corresponding angular (virtual) channel, $\boldsymbol{q}^a$, is a vector of the same size and can be derived using the DFT matrix $\boldsymbol{U_r}$ as $\boldsymbol{q}^a =\boldsymbol{U}^H_{\boldsymbol{r}} \: \boldsymbol{q}$ \cite{sayeed2002deconstructing} (this is merely a linear transformation that maps the sequence of channel gains into a sequence of gains from different AoAs. This mapping will be presented in more detail in Section \ref{SystemModel}).
Assume a single-path channel, i.e., the channel has only one cluster with a single path in it.
Let the path gain be denoted by $\alpha$. For simplicity assume $\alpha = 1$.
Further, let us assume perfect sparsity such that the AoA is along one of the directions defined in the DFT matrix $\boldsymbol{U_r}$, i.e., the channel path will only contribute to one angular bin.
Finally, let us also neglect the channel noise.

\begin{figure}[t]
\centering
\includegraphics[width=0.5\linewidth]{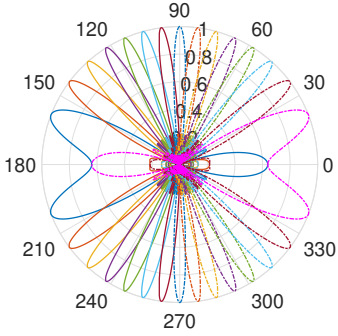}
\caption{\small Beam patterns of all possible angular directions}
\label{fig:AllBeams} 
\end{figure}

Based on the channel description above, we get an angular channel vector of the form
\begin{equation}
\boldsymbol{q}^a = 
 \begin{pmatrix}
  q^a_0 & q^a_1 &  \hdots &  q^a_{14}
 \end{pmatrix}^T,
\end{equation}
such that $q^a_i \in \{0,1\}$ and the number of non-zero elements in $\boldsymbol{q}^a$ is $1$. Any component of $\boldsymbol{q}^a$ can be measured using one of the beam patterns shown in Fig. \ref{fig:AllBeams}.

\textbf{Objective:}
Suppose the transmitter sends pilot symbols of the form $x {=} 1$. Thus, the received vector $\boldsymbol{y}$ of size $15 {\times} 1$ can be obtained as
\begin{equation}
\boldsymbol{y}   =  \boldsymbol{q} x   =  \boldsymbol{q} \Longleftrightarrow \boldsymbol{y}^a    =  \boldsymbol{q}^a
\end{equation}
where $\boldsymbol{y}^a$ is the received vector in the angular domain. So, with change of basis, we can think of $\boldsymbol{q}^a$ as a received sequence with just one non-zero component. To identify the position of this non-zero component, the receiver performs a sequence of channel measurements. Let $y_{s_i}$ denote the $i^{th}$ measurement such that
\begin{equation}
y_{s_i} = \boldsymbol{w}_i^H \boldsymbol{y} = \boldsymbol{w}_i^H \boldsymbol{q},
\end{equation}
where $\boldsymbol{w}_i$ denotes the $i^{th}$ receive (rx-)combining vector.

Our aim is to design channel measurements (i.e., $\boldsymbol{w_i}$'s) such that the correct AoA is identified using the minimum number of measurements.

\textbf{Proposed Solution:}
We consider this non-zero component to be an anomaly to a normally all-zero $15$-bin angular channel.
Hence, the goal of identifying its position is analogous to finding the most likely $1$-bit error pattern of a $15$-bit codeword in a linear block code.
Now, we need to identify an error correction code with codewords of length $15$ and with $1$-bit error correction capability \cite{van2012introduction}. Hence, we can use the binary $(15,11,3)$ Hamming code with parity check matrix $\boldsymbol{H}$ of size $4 {\times} 15$ and given by
\begin{equation}
  \boldsymbol{H}{=}
  \begin{pmatrix}
  1 & 0 & 0 & 0 & 1 & 0 & 0 & 1 & 1 & 0 & 1 & 0 & 1 & 1 & 1 \\
  0 & 1 & 0 & 0 & 1 & 1 & 0 & 1 & 0 & 1 & 1 & 1 & 1 & 0 & 0 \\
  0 & 0 & 1 & 0 & 0 & 1 & 1 & 0 & 1 & 0 & 1 & 1 & 1 & 1 & 0 \\
  0 & 0 & 0 & 1 & 0 & 0 & 1 & 1 & 0 & 1 & 0 & 1 & 1 & 1 & 1 \\
  \end{pmatrix}
\end{equation}
where $h_{i,j}$ represents the component at the intersection of row $i$ and column $j$ of $\boldsymbol{H}$. Using hard decision decoding of LBCs, error syndrome vectors of length $4$ are obtained. Every possible syndrome vector maps to only one correctable error pattern\footnote{A correctable error pattern of a ($15,11,3$) Hamming code is any $15{\times}1$ binary vector that contains only one '1' (at the error's position).}.
Similarly, for channel estimation, several measurements should be performed at RX where each measurement mimics the behavior of a corresponding element in the error syndrome vector. Each measurement boils down to adding signals from a subset of the available $15$ directions. Since each measurement can either include the direction of the incoming strong path of gain $\alpha=1$ or no strong paths at all, then the elements of the channel syndrome vector are in $\{0,1\}$.

\begin{figure*}[t]
\centering
\includegraphics[width=1\linewidth]{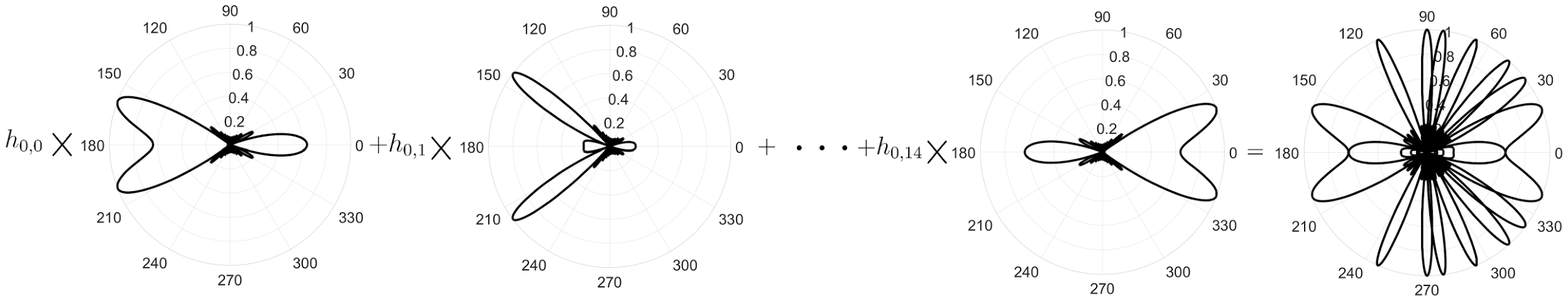}
\caption{\small Beam pattern of receive combining vector $\boldsymbol{w_0}$}
\label{fig:RxComb} 
\end{figure*}

For every measurement $y_{s_i}$, we design $\boldsymbol{w}_i$ based on the entries of the $i^{th}$ row of $\boldsymbol{H}$ such that: if $h_{i,j} = 1$, then we include the beam pattern that points to direction $j$ in $\boldsymbol{w}_i$.
For example, the $0^{th}$ row of $\boldsymbol{H}$ is given by $[1 \: 0 \: 0 \: 0 \: 1 \: 0 \: 0 \: 1 \: 1 \: 0 \: 1 \: 0 \: 1 \: 1 \: 1]$. Hence, $\boldsymbol{w}_0$ should include beam patterns pointing to the set of directions $\{0,4,7,8,10,12,13,14\}$.

Fig. \ref{fig:RxComb} illustrates this operation for $\boldsymbol{w}_0$. We can see that the resultant beam pattern of $\boldsymbol{w}_i$ combines signals coming from a set of selected directions dictated by the $i^{th}$ row of $\boldsymbol{H}$.
We call the obtained measurement vector, $\boldsymbol{y_s}$, the \textit{channel syndrome} which is analogous to error syndromes in hard decision decoding of LBCs.
Then, a table that maps every possible channel syndrome to a unique corresponding channel can be constructed.
Table \ref{table:SyndromeMap} shows this mapping.

\begin{table}[t]
\caption{\small Mapping of channel syndromes to angular channels}
\label{table:SyndromeMap}
\centering
 \begin{tabular}{||c | c||} 
 \hline
 Channel Syndrome $\boldsymbol{y}^T_{\boldsymbol{s}}$ & Angular Channel ${\boldsymbol{q}^a}^T$ \Tstrut \Bstrut \\ 
 \hline\hline
$[0 \: 0 \: 0 \: 0]$ & $[0 \: 0 \: 0 \: 0 \: 0 \: 0 \: 0 \: 0 \: 0 \: 0 \: 0 \: 0 \: 0 \: 0 \: 0]$\\ 
 \hline
$[1 \: 0 \: 0 \: 0]$ & $[1 \: 0 \: 0 \: 0 \: 0 \: 0 \: 0 \: 0 \: 0 \: 0 \: 0 \: 0 \: 0 \: 0 \: 0]$\\
 \hline
$[0 \: 1 \: 0 \: 0]$ & $[0 \: 1 \: 0 \: 0 \: 0 \: 0 \: 0 \: 0 \: 0 \: 0 \: 0 \: 0 \: 0 \: 0 \: 0]$\\
 \hline
$[0 \: 0 \: 1 \: 0]$ & $[0 \: 0 \: 1 \: 0 \: 0 \: 0 \: 0 \: 0 \: 0 \: 0 \: 0 \: 0 \: 0 \: 0 \: 0]$\\
 \hline
$[0 \: 0 \: 0 \: 1]$ & $[0 \: 0 \: 0 \: 1 \: 0 \: 0 \: 0 \: 0 \: 0 \: 0 \: 0 \: 0 \: 0 \: 0 \: 0]$\\
 \hline
$[1 \: 1 \: 0 \: 0]$ & $[0 \: 0 \: 0 \: 0 \: 1 \: 0 \: 0 \: 0 \: 0 \: 0 \: 0 \: 0 \: 0 \: 0 \: 0]$\\
 \hline
$[0 \: 1 \: 1 \: 0]$ & $[0 \: 0 \: 0 \: 0 \: 0 \: 1 \: 0 \: 0 \: 0 \: 0 \: 0 \: 0 \: 0 \: 0 \: 0]$\\
 \hline
$[0 \: 0 \: 1 \: 1]$ & $[0 \: 0 \: 0 \: 0 \: 0 \: 0 \: 1 \: 0 \: 0 \: 0 \: 0 \: 0 \: 0 \: 0 \: 0]$\\
 \hline
$[1 \: 1 \: 0 \: 1]$ & $[0 \: 0 \: 0 \: 0 \: 0 \: 0 \: 0 \: 1 \: 0 \: 0 \: 0 \: 0 \: 0 \: 0 \: 0]$\\
 \hline
$[1 \: 0 \: 1 \: 0]$ & $[0 \: 0 \: 0 \: 0 \: 0 \: 0 \: 0 \: 0 \: 1 \: 0 \: 0 \: 0 \: 0 \: 0 \: 0]$\\
 \hline
$[0 \: 1 \: 0 \: 1]$ & $[0 \: 0 \: 0 \: 0 \: 0 \: 0 \: 0 \: 0 \: 0 \: 1 \: 0 \: 0 \: 0 \: 0 \: 0]$\\
 \hline
$[1 \: 1 \: 1 \: 0]$ & $[0 \: 0 \: 0 \: 0 \: 0 \: 0 \: 0 \: 0 \: 0 \: 0 \: 1 \: 0 \: 0 \: 0 \: 0]$\\
 \hline
$[0 \: 1 \: 1 \: 1]$ & $[0 \: 0 \: 0 \: 0 \: 0 \: 0 \: 0 \: 0 \: 0 \: 0 \: 0 \: 1 \: 0 \: 0 \: 0]$\\
 \hline
$[1 \: 1 \: 1 \: 1]$ & $[0 \: 0 \: 0 \: 0 \: 0 \: 0 \: 0 \: 0 \: 0 \: 0 \: 0 \: 0 \: 1 \: 0 \: 0]$\\
 \hline
$[1 \: 0 \: 1 \: 1]$ & $[0 \: 0 \: 0 \: 0 \: 0 \: 0 \: 0 \: 0 \: 0 \: 0 \: 0 \: 0 \: 0 \: 1 \: 0]$\\
 \hline
$[1 \: 0 \: 0 \: 1]$ & $[0 \: 0 \: 0 \: 0 \: 0 \: 0 \: 0 \: 0 \: 0 \: 0 \: 0 \: 0 \: 0 \: 0 \: 1]$\\ 
\hline
\end{tabular}
\end{table}

In this example, we are able to estimate the channel based on only $4$ measurements as opposed to $15$, which is the number of measurements with exhaustive search.
Important aspects of our proposed method include the choice of codes, the design of precoding and rx-combining measurement vectors, the effect of variable gains and phases of different paths and the occurrence of measurement errors.

\begin{remark}[Receiver Architecture] Note that, to achieve beam patterns similar to the one shown in Fig. \ref{fig:RxComb}, the receiver architecture needs to be a bit different from those of classical analog/hybrid beamforming architectures. Specifically, in addition to low-noise amplifiers (LNA) typically placed at the output of each antenna, we will need to add controllable low-power amplifiers, as well. The resultant architecture is still quite simple (see Fig. \ref{fig:architecture}). That is, besides the low-power amplifiers, the proposed architecture is similar to those of simple analog beamforming. Moreover, relatively low-resolution ADCs can be used which mitigates the high power consumption problem associated with high-resolution ADCs.
\end{remark}

\textbf{Motivation for LBC-inspired approach:}
LBCs are designed to discover and correct a certain maximum number of errors in a codeword of a specified length. This objective is achieved by adding redundant \textit{parity check} bits to the original information sequence.
What makes our devised approach attractive is that the number of measurements needed for channel estimation can be shown to be equal to the number of parity bits of some corresponding code. Hence, we can control the estimation performance via appropriate code selection.
In this work, we will propose a method to specify the number of necessary channel measurements as a function of the rate of the underlying code.


\section{System Model}
\label{SystemModel}

\begin{figure}[t]
\centering
\includegraphics[width=1\linewidth]{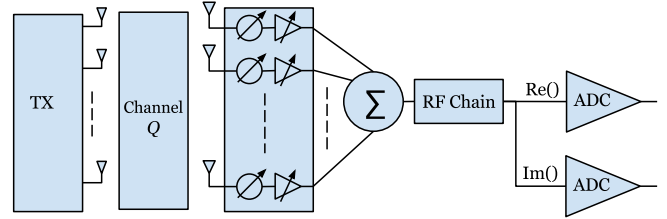}
\caption{\small Hardware Block Diagram: Every antenna is connected to a phase shifter and low-power variable gain amplifier. Then, all outputs are combined using an adder and passed to an RF chain with in-phase and quadrature channels.}
\label{fig:architecture} 
\end{figure}

Consider a point-to-point millimeter-wave wireless communication system with a transmitter (TX) equipped with $n_t$ antennas and a receiver (RX) with $n_r$ antennas placed at fixed locations. Uniform Linear Arrays (ULA) are assumed at both TX and RX where each antenna element is connected to a phase shifter and a variable gain amplifier. A single RF chain at the receiver, with in-phase (I) and quadrature (Q) channels, is fed through a linear combiner (see Fig. \ref{fig:architecture}). Only two mid-tread ADCs, with $2^b{+}1$ quantization levels, are utilized, where quantization levels take values from the set $\mathcal{Y} = \{-2^{b-1}, \dots,-1,0,1,\dots,2^{b-1}\}$.

We adopt a single-tap channel model where $\boldsymbol{Q} \in \mathbb{C}^{ n_r { \times } n_t}$ denotes the channel matrix between TX and RX.
Assume that the channel has $L$ clusters, where each cluster contains a single path with gain $\alpha_l$, AoD $\theta_l$, and AoA $\phi_l$.
The channel is assumed to be sparse such that $L \ll n_t, n_r$.
Let $\alpha_l^b \in \mathbb{C}$ denote the baseband channel gain and is defined as
\begin{equation}
\alpha_l^b {=} \alpha_l \: \sqrt{n_t n_r} \: e^{- j\frac{2 \pi \rho_l}{\lambda_c}}
\end{equation}
where $\rho_l$ is the length of path $l$ and $\lambda_c$ is the carrier wavelength.
The angular cosines of AoD and AoA associated with path $l$ are denoted by $\Omega_{tl}$ and $\Omega_{rl}$, respectively. The transmit and receive spatial signatures along the direction $\Omega$ are given by $\boldsymbol{e_t}(\Omega)$ and $\boldsymbol{e_r}(\Omega)$ such that
\begin{equation}
\boldsymbol{e_t}(\Omega) {=} \frac{1}{\sqrt{n_t}}
  \begin{pmatrix}
  1    							    	  \\
  e^{-j2\pi   \Delta_t \Omega}  	  \\
  e^{-j2\pi 2 \Delta_t \Omega}  	  \\
  \vdots								  \\
  e^{-j2\pi (n_t-1) \Delta_t \Omega}  \\
  \end{pmatrix},
\end{equation}
where  $\boldsymbol{e_r}(\Omega)$ has a similar definition to $\boldsymbol{e_t}(\Omega)$, and
$\Delta_t$ and $\Delta_r$ are the antenna separations at TX and RX normalized by the wavelength $\lambda_c$.
Let the average path loss be denoted by $\mu$.
Thus, 
$\boldsymbol{Q}$ can be written as
\begin{equation}
\boldsymbol{Q} = \sum_{l=1}^{L} \frac{\alpha_l^b}{\mu}  \: \boldsymbol{e_r}(\Omega_{rl}) \: \boldsymbol{e}^H_{\boldsymbol{t}}(\Omega_{tl}).
\end{equation}
We define $\boldsymbol{U_t}$ and $\boldsymbol{U_r}$ as the unitary Discrete Fourier Transform (DFT) matrices whose columns constitute an orthonormal basis for the transmit and receive signal spaces $\mathbb{C}^{n_t}$ and $\mathbb{C}^{n_r}$, respectively. $\boldsymbol{U_t}$ (and similarly $\boldsymbol{U_r}$) is given by
\begin{equation}
\boldsymbol{U_t} = 
  \begin{pmatrix}
  \boldsymbol{e_t}(0) & \boldsymbol{e_t}(\frac{1}{L_t}) & \dots & \boldsymbol{e_t}(\frac{n_t-1}{L_t})
  \end{pmatrix},
\end{equation}
where $L_t$ and $L_r$ are the normalized lengths of the transmit and receive antenna arrays such that $L_t = n_t \Delta_t$ and $L_r = n_r \Delta_r$.
Let $\boldsymbol{Q^a}$ be the channel matrix in the angular domain \cite{sayeed2002deconstructing}, where
\begin{equation} \label{AngularChannel}
\boldsymbol{Q^a} = \boldsymbol{U}^H_{\boldsymbol{r}} \boldsymbol{Q} \boldsymbol{U_t}.
\end{equation}
The rows and columns of $\boldsymbol{Q}^a$ divide the channel into resolvable RX and TX bins, respectively.
Further, we assume a perfect sparsity model in which AoDs $\theta_l$, and AoA $\phi_l$, are along the directions defined in $\boldsymbol{U_t}$ and $\boldsymbol{U_r}$ \cite{ChEstScniter14,alkhateeb2014channel,RACE_2017}. Hence, each channel cluster will only contribute to a single pair of TX and RX bins. Therefore, $\boldsymbol{Q}^a$ has a maximum of L non-zero TX and RX bins.

The baseband channel model is given by
\begin{equation}\label{eqn:baseBandModel}
\boldsymbol{y} = \boldsymbol{Q} \boldsymbol{x} + \boldsymbol{n},
\end{equation}
where $\boldsymbol{x} \in \mathbb{C}^{n_t}$ is the transmitted signal, $\boldsymbol{y} \in \mathbb{C}^{n_r}$ is the received signal and $\boldsymbol{n} \sim \mathcal{CN}(\boldsymbol{0} , N_0 \boldsymbol{I}_{n_r} )$
is an i.i.d. complex Gaussian noise vector.

Let $\boldsymbol{f} \in \mathbb{C}^{n_t} $ and $\boldsymbol{w} \in \mathbb{C}^{n_r}$ be the precoding and rx-combining vectors, respectively. The transmit signal $\boldsymbol{x}$ is given by $\boldsymbol{x} = \boldsymbol{f} s$ where $s$ is the transmitted symbol with average power $\mathbb{E}(ss^H) {=} P$. After the receiver applies the rx-combining vector $\boldsymbol{w}$, the resultant symbol $u$ can be written as
\begin{equation}\label{eqn:rxSymbol}
u = \boldsymbol{w}^H \boldsymbol{Q} \boldsymbol{f} s + \boldsymbol{w}^H \boldsymbol{n}.
\end{equation}
Afterwards, $u$ is passed forward to the ADCs. There, a quantized version, $u_s$, of $u$ is obtained such that
\begin{equation}\label{eqn:quantizedRxSymbol}
u_s = [\boldsymbol{w}^H \boldsymbol{Q} \boldsymbol{f} s + \boldsymbol{w}^H \boldsymbol{n}]_{+} \: ,
\end{equation}
where $[\: \cdot \:]_+$ represents the quantization function. Now, $u_s$ constitutes a single, quantized, unit measurement obtained using specific $\boldsymbol{f}$ and $\boldsymbol{w}$ vectors such that $u_s \in \mathcal{Y}^2 \subset \mathbb{C}$.

We assume that $\boldsymbol{Q}$ remains fixed throughout the entire estimation process. The noise component $\boldsymbol{w}^H \boldsymbol{n}$ normalized by $\left\Vert \boldsymbol{w} \right\Vert$ is also a complex gaussian random variable such that $\frac{\boldsymbol{w}^H}{\left\Vert \boldsymbol{w} \right\Vert} \boldsymbol{n} \sim \mathcal{CN}(0,N_0)$.
We define the signal to noise ratio (\textit{SNR}) on a per path basis such that \textit{SNR} of path $l$ is given by
\begin{equation}
\label{eqn:SNRl}
\textit{SNR}_l = \frac{P}{N_0} \left\lvert \frac{\alpha_l^b}{\mu} \right\rvert^2.
\end{equation}
Note that the actual received \textit{SNR} depends on all path gains included in a measurement.

\section{Problem Statement}
\label{ProblemStatement}

Suppose a maximum number of $L$ clusters need to be discovered in the channel where $L \ll n_t,n_r$. Under the prefect sparsity assumption, $\boldsymbol{Q}^a$ has a maximum of $L$ non-zero RX and TX angular bins.
Our objective is to identify the angular positions at which channel clusters exist and identify their path gain values using the \textbf{least possible number of measurements}. Let the number of measurements be $m$ such that each measurement, $u_{s_{i,j}}$, is obtained using the precoder $\boldsymbol{f_j}$ and rx-combiner $\boldsymbol{w_i}$. Let the number of rx-combiners and precoders be $m_1$ and $m_2$, respectively. Measurements take the form $u_{s_{i,j}} = [\boldsymbol{w_i}^H \boldsymbol{Q} \boldsymbol{f_j} s + \boldsymbol{w_i}^H \boldsymbol{n}]_+$. 
Let $\xi()$ be a mapping function that takes in the measurements $\{u_{s_{i,j}}\}_{\forall i,j}$ as inputs and returns the estimated channel $\widehat{\boldsymbol{Q}}^{\boldsymbol{a}}$.
For each $j$, we stack the measurements $\{u_{s_{i,j}}\}_{\forall i}$ in a single (syndrome) vector such that $\boldsymbol{u_{s_j}} = [u_{s_{0,j}} \: u_{s_{1,j}} \dots u_{s_{{{m_1}-1},j}}]^T$.
Our design variables are the precoding vectors $\boldsymbol{f_j}$, rx-combining vectors $\boldsymbol{w_i}$, the number of measurements $m$, the mapping function $\xi()$, and the transmitted symbol power $P$.

In its essence, solving this problem boils down to finding the optimal set of measurements $\{u_{s_{i,j}}\}_{\forall i,j}$ and the mapping function $\xi()$ such that $\boldsymbol{Q}^a$ can be estimated using the minimum number of measurements.
For ease of explanation, we first consider a channel with a single transmit antenna and $n_r$ receive antennas. Therefore, no precoding is needed and the design of measurements is reduced to designing the rx-combining vectors $\boldsymbol{w_i}$.
Recall that in the motivating example in Section \ref{illustrative}, we dealt with a special case of $n_r {\times} 1$ channels where we sought to find the direction of arrival of a channel with a single path of known gain, $\alpha = 1$. In the general case, we should consider arbitrary path gains $\alpha \in \mathbb{C}$ and channels with multiple paths.

\section{Beam Discovery}
\label{BeamDiscovery}
In this section, we present our proposed solution. As an initial step, we solve a simplified version of the problem where communication channels have a single transmit antenna and multiple receive antennas. Afterwards, we will build on it to provide the solution for general channels with multiple transmit and receive antennas.

\subsection{Beam Detection using LBC-inspired approach}
\label{LBC_inspired}
To identify the exact number of measurements and their corresponding design, we follow a decoding-like approach of
LBC\footnote{\label{LBC}In channel coding, the convention is to use row vectors. Thus, let $\boldsymbol{x}$ and $\boldsymbol{c}$ be $1{\times}k$ and $1{\times}n$ binary row vectors that represent an information sequence and its corresponding codeword of an LBC, respectively. Also let $\boldsymbol{r} {=} \boldsymbol{c} {+} \boldsymbol{e}$ be a received sequence corrupted by $1{\times}n$ error pattern $\boldsymbol{e}$. To decode $\boldsymbol{r}$, we calculate an error syndrome  vector $\boldsymbol{s}$, of size $1{\times}n{-}k$, such that $\boldsymbol{s} {=} \boldsymbol{r} \boldsymbol{H}^T$, where $\boldsymbol{H}$ is the parity check matrix of the used LBC. Then, a most likely error pattern $\hat{\boldsymbol{e}}$ can be uniquely identified by $\boldsymbol{s}$ using a look-up table called the \textit{standard array}. Finally, the decoded codeword is obtained using $\hat{\boldsymbol{c}}{=}\boldsymbol{r}{-}\hat{\boldsymbol{e}}$. A decoding error occurs if the number of errors, identified using $1$'s in $\boldsymbol{e}$, is beyond the error correction capability of the used code, denoted by $e_n$.
Note that in this context, all vectors, matrices and math operations are over GF(2).}.
First, we need to find an LBC, $C$, that has an error correction capability $e_n$ such that i) the maximum number of clusters in the channel, $L$, is equal to $e_n$ and ii) the length of its codewords $n$ is equal to the number of antennas $n_r$ ($n_r $ is also the number of resolvable directions).
The code $C$ has a parity check matrix $\boldsymbol{H}$ which represents the link between channel decoding and beam detection problems.
Binary codes deal with data and error sequences defined over the finite field $GF(2)$, i.e., addition and multiplication operations are defined over $GF(2)$ with binary inputs and outputs, i.e., $1$'s and $0$'s.
However, mm-wave channel parameters are defined over the complex numbers field $\mathbb{C}$. Therefore, to account for arbitrary path gains, we should be able to extend this concept to $\mathbb{C}$.

\textbf{Although $\boldsymbol{H}$ is defined over $GF(2)$, we interpret its $'1'$ and $'0'$ entries as real numbers.}
Then, similar to channel decoding, we seek to obtain a channel syndrome, $\boldsymbol{y_s}$, such that
${(\boldsymbol{y}_s)}^T {=} {(\boldsymbol{q}^a)}^T \boldsymbol{H}^T \Longrightarrow \boldsymbol{y}_s {=} \boldsymbol{H} \boldsymbol{q}^a$.
This matrix multiplication can be realized using channel measurements such that each measurement gives one component in $\boldsymbol{y_s}$.
Measurements $\{ y_{s_i}\}_{\forall i}$ make up the components of the channel syndrome vector $\boldsymbol{y_{s}}$. Then, we need to find a mapping function $\xi()$ that takes in the channel syndrome vector $\{\boldsymbol{y_{s}}\}$ as an input and returns the estimated channel $\widehat{\boldsymbol{q}}^{\boldsymbol{a}}$.
The position of each non-zero component in $\widehat{\boldsymbol{q}}^{\boldsymbol{a}}$ identifies a path's AoA, and its value identifies its baseband path gain.
Finally, for this to work, we need to show that such channel measurements provide one-to-one mapping to the channel. In other words, $\boldsymbol{y_s}$ must be a \textit{sufficient statistic} for estimating the channel.
In Section \ref{SufficientStatistic}, we will show that our design results in the sufficient statistic we seek to achieve.

\begin{remark}[Difference between $\boldsymbol{y_s}$ and $\boldsymbol{u_s}$]
Both $\boldsymbol{y_s}$ and $\boldsymbol{u_s}$ refer to vectors of measurement symbols, however, $\boldsymbol{u_s}$ is considered to be the noise corrupted and quantized version of $\boldsymbol{y_s}$.
Specifically, $\boldsymbol{u_s} {=} [\boldsymbol{y_s} + \boldsymbol{z}]_+ $ such that $\boldsymbol{z}$ is the measurement noise vector.
While $\boldsymbol{u_s}$ is what we expect to observe, our design of measurements focuses on finding $\boldsymbol{y_s}$; an error-free symbol.
Of course, errors degrade beam discovery performance. Thus, in Section \ref{ErrorCorr}, we will deal with the effect of measurement errors separately and present a solution that increases reliability of beam discovery.
The separate treatment of measurement errors simplifies the design and provides a clear understanding of the nature of our solution.

\end{remark}

\begin{remark}[Number of Measurements]
The solution we obtain is dependent on channel parameters, namely, the number of antennas and the sparsity level of the channel.
That is, at a fixed sparsity level, i.e., fixed number of clusters $L$, a larger number of antennas necessitates more channel measurements. In other other words, the high resolution realized by large $n_r$ comes at a price of an increased number of measurements.
Similarly, at fixed $n_r$, more channel clusters involve more measurements for correct channel estimation.
\end{remark}

\subsection{Measurements Design}
\label{MeasurementsDesign}
Recall that each component in $\boldsymbol{q}^a$ represents a resolvable angular direction at the receiver. Let each resolvable direction be given an identification number ($\textit{dir}_{\textit{rx}}$\textit{\#i}).
Also let $\textit{beam}_{\textit{rx}}$\textit{\#i} denote the beam pattern pointing to $\textit{dir}_{\textit{rx}}$\textit{\#i}, i.e., a signal coming from $\textit{dir}_{\textit{rx}}$\textit{\#i} can be individually measured using $\textit{beam}_{\textit{rx}}$\textit{\#i} (similar to beam patterns in Fig. \ref{fig:AllBeams}).

Now, we seek to obtain $\boldsymbol{y_s} {=} \boldsymbol{H} \boldsymbol{q^a}$ using careful design of $\boldsymbol{w_i}$'s
, i.e.,
\begin{equation}\label{syndromeEqn}
\boldsymbol{y_s} = 
 \begin{pmatrix}
  y_{s_0} \\
  y_{s_1} \\
  \vdots\\
  y_{s_{m{-}1}}
 \end{pmatrix}
 =
  \begin{pmatrix}
    \boldsymbol{w}^H_{\boldsymbol{0}} \boldsymbol{q}\\
    \boldsymbol{w}^H_{\boldsymbol{1}} \boldsymbol{q}\\
    \vdots\\
    \boldsymbol{w}^H_{\boldsymbol{m{-}1}} \boldsymbol{q}\\
  \end{pmatrix}
  \equiv
  \boldsymbol{H} \boldsymbol{q^a}.
\end{equation}
To achieve this, each rx-combining vector $\boldsymbol{w_i}$ is designed as a multi-armed beam, i.e., composed of several sub-beams similar to the beam pattern in Fig. \ref{fig:RxComb}. The sub-beams included in each $\boldsymbol{w_i}$ are identified by the $i^{th}$ row of the matrix $\boldsymbol{H}$.
That is, only if $h_{i,j}$, the intersection of the $i^{th}$ row and $j^{th}$ column, is $=1$, do we include $\textit{beam}_{\textit{rx}}$\textit{\#j} as a sub-beam in $\boldsymbol{w_i}$ (also refer to our discussion in Section \ref{illustrative}).


The design of rx-combining vectors is a crucial aspect of this work.
As an initial step towards obtaining proper rx-combining vectors, we consider designing $\boldsymbol{w_i}$'s using linear summation of all analog beamformers that correspond to \textit{beam}$_{\text{\textit{rx}}}$\textit{\#j}'s $\forall j : h_{i,j} = 1$.
Let $\Omega_j {=} \cos(\phi_j) {=} \frac{j}{L_r} \: \forall j {\in} \{0,\dots,n_r{-}1\}$, such that $\boldsymbol{e_r}(\Omega_j)$ is the spatial signature of $\textit{beam}_{\textit{rx}}$\textit{\#j}.
Then, $\boldsymbol{w_i}$ can be designed as
\begin{equation}
\boldsymbol{w_i} = \sum_{j=0}^{n_r-1} \mathbbm{1}_{ \{ h_{i,j}  = 1 \} } \boldsymbol{e_r}(\Omega_j)
\end{equation}

\subsection{Sufficient Statistic}
\label{SufficientStatistic}

We will show in this section that each channel syndrome can only be mapped to a single \textit{measurable channel}. A measurable channel in this context refers to $n_r{\times}1$ channels with $L$ non-zero components such that $L \leq e_n$, where $e_n$ is the error correction capability of the underlying code $C$ and $n_r {=} n$ is its CWs length.
Let $\mathcal{Q}^a$ be the set of all measurable channels:
\begin{equation}\label{eqn:setOfMeasurableChannels}
\mathcal{Q}^a \triangleq  \left\lbrace \boldsymbol{q^a} \in \mathbb{C}^{n_r} : \left\vert q^a_i {:} q^a_i {\neq} 0 \right\vert \leq e_n  \right\rbrace.
\end{equation}
%
Since each measurement combines signals coming from multiple directions, each element in the channel syndrome vector is a linear combination of a subset of the available paths. In other words, each measurement has the possibility that one or more paths are included in it.
This setting is rather challenging. To understand why, consider a channel that has two paths with gains $\alpha_1,\alpha_2 \in \mathbb{C}$. Suppose that
$\alpha_1$ and $\alpha_2$ are of equal magnitudes but are out-of-phase (i.e., phase shift $= 180^\circ$). Hence, if signals coming from both paths are combined in a single measurement, the resultant value is $0$ which is similar to the result we get if no paths exist in the measured directions.
Also each channel measurement can be a result of endless possibilities for the combined path gain values.
So, a natural question to ask is: does this ambiguity cause measurement errors?
The direct answer to this question is: \textbf{No}.
In the sequel we will show that the resulting channel syndrome, i.e., the combination of all channel measurements, is sufficient to correctly estimate the channel.

First, recall our discussion in Footnote \ref{LBC}.
Then, consider all \textit{single-bit error patterns} $\boldsymbol{e^{(i)}}$ 
of a code $C$, with maximum number of correctable errors ${=}e_n$, such that
\begin{equation}
\label{eqn:singleBitErrors}
e_k^{(i)} = \begin{cases}
               1, & k = i\\
               0, & k \neq i
            \end{cases}
\end{equation}
where $e_k^{(i)}$ is the $k^{th}$ component of $\boldsymbol{e^{(i)}}$.
Also let $\boldsymbol{s^{(i)}}$ be the corresponding error syndrome of $\boldsymbol{e^{(i)}}$.
Recall that $\boldsymbol{s^{(i)}} {=} \boldsymbol{e^{(i)}}\boldsymbol{H}^T$.
Hence, we can see that $\boldsymbol{s^{(i)}}$ is exactly the $i^{th}$ row of $\boldsymbol{H}^T$, i.e., $i^{th}$ column of $\boldsymbol{H}$.
Let $\mathcal{E}_C$ denote the set of correctable error patterns of the code $C:$
\begin{multline}\label{eqn:setOfCorrectableErrors}
\mathcal{E}_C \triangleq  \Big\{ \boldsymbol{e} {\in} \{0,1\}^n : \boldsymbol{e} {=} \sum_{i=1}^{n} \omega_i \boldsymbol{e}^{(i)} , \\
\omega_i \in \{0,1\} : \left\vert \omega_i {:} \omega_i {=} 1 \right\vert \leq e_n  \Big\}.
\end{multline}
Now, we can write any correctable error pattern $\boldsymbol{e} \in \mathcal{E}_C$ as a linear combination of all single-bit error patterns over the finite field $GF(2)$ such that
\begin{equation}
\label{eqn:errorPattern}
\boldsymbol{e} = \omega_1 \boldsymbol{e}^{(1)} + \omega_2 \boldsymbol{e}^{(2)} + \dots + \omega_n \boldsymbol{e}^{(n)}
\end{equation}
and its corresponding error syndrome is
\begin{equation}
\boldsymbol{s} = \omega_1 \boldsymbol{s}^{(1)} + \omega_2 \boldsymbol{s}^{(2)} + \dots + \omega_n \boldsymbol{s}^{(n)}
\end{equation}
\begin{lemma}
For an error pattern $\boldsymbol{e}_t$ with number of bit errors identical to $e_n$, its syndrome $\boldsymbol{s}_t$ is a linear combination of $e_n$ linearly independent vectors $\boldsymbol{s}^{(i)}$.
\end{lemma}
\begin{proof}
We are going to prove this lemma by contradiction.
First, assume that $\boldsymbol{s}_t$ is a linear combination of $e_n$ linearly \textbf{dependent} vectors $\boldsymbol{s}^{(i)}$ over $GF(2)$.
Therefore, there exists another error syndrome $\boldsymbol{s}_t^*$ composed of only linear combination of independent vectors $\boldsymbol{s}^{(i)}$ such that $\boldsymbol{s}_t = \boldsymbol{s}_t^*$.
Therefore, there exists another error patter $\boldsymbol{e}_t^*$ with number of errors strictly less than $e_n$ such that its syndrome $\boldsymbol{s}_t^* = \boldsymbol{s}_t$. Since $\boldsymbol{e}_t^*$ has a number of errors less than $e_n$, then it is a correctable error pattern, and since all error syndromes of correctable error patterns are different, then $\boldsymbol{s}_t^*$ should be $\neq \boldsymbol{s}_t$. Hence, we arrive at a contradiction.
\end{proof}
It is also easy to see that if $\boldsymbol{e}_{t_1}$ and $\boldsymbol{e}_{t_2}$ are two different correctable error patterns, then their error syndromes $\boldsymbol{s}_{t_1}$ and $\boldsymbol{s}_{t_2}$ are composed of a linear combination of different sets of single-bit error syndromes $\boldsymbol{s}^{(i)}$.

\begin{lemma}\label{lemma:LinInd}
Any $n-$dimensional linearly independent vectors over $GF(2)$, are also linearly independent over $\mathbb{C}^n$.
\end{lemma}
\begin{proof}
Let $\boldsymbol{v_1}, \dots, \boldsymbol{v_m}$ be a set of $n-$dimensional vectors defined over $GF(2)$. The vectors $\boldsymbol{v_i}$ can be made the columns of an $n {\times} m$ matrix $\boldsymbol{\Psi}$. Since all $\boldsymbol{v_i}$'s are linearly independent over $GF(2)$, then $\boldsymbol{\Psi}$ is a left-invertible matrix. Therefore, there exists a non-zero (modulo $2$) $m {\times} m$ minor of $\boldsymbol{\Psi}$. Now, suppose the entries in $\boldsymbol{\Psi}$ are interpreted as real numbers. Therefore, $\boldsymbol{\Psi}$, now taken over $\mathbb{R}$, has an $m {\times} m$ sub-matrix whose determinant is non-zero, which proves that it is invertible. Therefore, the vectors $\boldsymbol{v_i}$'s, i.e., columns of $\boldsymbol{\Psi}$, are linearly independent over $\mathbb{R}$ which, using the same argument, can also be shown to be linearly independent over $\mathbb{C}$.
\end{proof}

Suppose that entries of $\boldsymbol{H}$ and $\boldsymbol{e^{(i)}}$ are interpreted as real numbers, then we can write the channel $\boldsymbol{q^a}$ as
\begin{equation}
({\boldsymbol{q}^a})^T = \alpha_1 \boldsymbol{e}^{(1)} + \alpha_2 \boldsymbol{e}^{(2)} + \dots + \alpha_n \boldsymbol{e}^{(n)}
\end{equation}
where $\alpha_i {\in} \mathbb{C}$
and $\sum_{i=1}^n \mathbb{\mathbbm{1}}_{\{\alpha_i {\neq} 0\}} {\leq} e_n$.
Therefore, each channel syndrome ${(\boldsymbol{y}_s)}^T {=} {(\boldsymbol{q}^a)}^T \boldsymbol{H}^T \Longrightarrow \boldsymbol{y}_s {=} \boldsymbol{H} \boldsymbol{q}^a$ is a linear combination of independent vectors in $\mathbb{C}^{n{-}k}$ (columns of $\boldsymbol{H}$). Therefore, all possible measurable channels yield unique channel syndromes which implies that they are sufficient for the channel estimation problem.

\subsection{Mapping Function $\xi()$}
\label{Mapping}
Now that we have shown that each measurable channel can be mapped to a unique channel syndrome, we need to find this mapping function, i.e., $\xi {:} \boldsymbol{y_s} \rightarrow \widehat{\boldsymbol{q}}^{\boldsymbol{a}}$, where $\widehat{\boldsymbol{q}}^{\boldsymbol{a}}$ denotes the estimated channel.
Next, we propose two different approaches to find $\xi()$.

\subsubsection{Look-up Table Method}
\label{subsub:LookUpMethod}
Again, we resolve to a technique used in hard decision decoding where a look-up table is constructed that maps every error syndrome to a corresponding error pattern. Likewise, we construct a look-up table that indicates which channel corresponds to an obtained channel syndrome.

Since we employ ADCs with finite resolution,
only a finite number of realizable syndromes, $\boldsymbol{y_s}$, exist (and a finite number of corresponding channels).
Therefore, a look-up table method is feasible. We construct the table by, first, generating all possible sparse angular channels.
Then, we find the corresponding channel syndromes using $\boldsymbol{y_s} = \boldsymbol{H} \boldsymbol{q}^a$, where $\boldsymbol{q}^a \in \mathcal{Q}^a$. 
Recall that the actual, noise-corrupted, received channel syndrome is $\boldsymbol{u_s}$. Therefore, $\boldsymbol{u_s}$ might not exactly match one of the channel syndrome vectors in the look-up table. Hence, we instead search for the $\boldsymbol{y_s}$ table entry that has the closest \textit{distance} $\delta$ to $\boldsymbol{u_s}$, and pick its corresponding channel as the estimated channel $\widehat{\boldsymbol{q}}^{\boldsymbol{a}}$.
We define the distance between the two complex $m-$dimensional vectors $\boldsymbol{y_s},\boldsymbol{u_s}$, to be the $l^2-$norm as follows:
\begin{equation}
\label{eqn:distance}
\delta(\boldsymbol{y_s},\boldsymbol{u_s}) = \norm{\boldsymbol{y_s}-\boldsymbol{u_s}}_2
= \sqrt{\sum_{i=0}^{m-1} \lvert y_{s_i} {-} u_{s_i} \rvert ^2 }.
\end{equation}

By obtaining $\widehat{\boldsymbol{q}}^{\boldsymbol{a}}$, we not only identify the AoA at the Rx, but we also obtain the magnitude and phase information associated with every strong path to the TX.

\begin{remark}[Size of look-up table]
The size of the look-up table scales proportionally with ADC resolution.
As the resolution of ADCs increases, the size of $\mathcal{Q}^a$ increases as well; since every non-zero component of every $\boldsymbol{q^a} {\in} \mathcal{Q}^a$ can take more values.
If low resolution ADCs could be tolerated, then the look-up table is a plausible choice for mapping owing to the small size of the look-up table.
However, if high resolution ADCs are needed, the look-up table size creates a problem for memory-limited devices especially for large antenna arrays. It also increases the complexity of finding the closest $\boldsymbol{y_s}$ table entry to the measurement $\boldsymbol{u_s}$.
Next, we will propose a different approach for mapping called \textit{search method} which does not scale with ADC resolution.
\end{remark}



\subsubsection{Search Method}
\label{subsub:SearchMethod}
Recall that $\boldsymbol{y_s} {=} \boldsymbol{H}\boldsymbol{q^a}$ (Eq. (\ref{syndromeEqn})), and let the parity check matrix $\boldsymbol{H}$ be represented as:
\begin{equation}
\boldsymbol{H} =
\begin{pmatrix}
\boldsymbol{h_1} & \boldsymbol{h_2} & \dots & \boldsymbol{h_{n_r}}
\end{pmatrix},
\end{equation}
where $\boldsymbol{h_i}$ is the $i^{th}$ column of $\boldsymbol{H}$.
Thus, we can write $\boldsymbol{y_s}$ as:
\begin{equation}
\boldsymbol{y_s} = q^a_1\boldsymbol{h_1} + q^a_2\boldsymbol{h_2} + \dots + q^a_{n_r}\boldsymbol{h_{n_r}},
\end{equation}
where $q^a_i$ is the $i^{th}$ component of $\boldsymbol{q^a}$.
Note that $\boldsymbol{q^a}$ is $L-$sparse, i.e., we have no more than $L$ non-zero components $q^a_i$.
Let the indices of the non-zero components be $x_1,x_2, \dots, x_L$, hence, $\boldsymbol{y_s}$ can succinctly be written as:
\begin{equation} \label{syndrome_short}
\boldsymbol{y_s} = q^a_{x_1}\boldsymbol{h_{x_1}} + q^a_{x_2}\boldsymbol{h_{x_2}} + \dots + q^a_{x_L}\boldsymbol{h_{x_L}}.
\end{equation}
Let
\begin{equation}
\label{eqn:EffC}
\boldsymbol{C} \triangleq
\begin{pmatrix}
\boldsymbol{h_{x_1}} & \boldsymbol{h_{x_2}} & \dots & \boldsymbol{h_{x_L}}
\end{pmatrix}.
\end{equation}
Then we can write Eq. (\ref{syndrome_short}) in matrix form as:
\begin{equation}
\label{eqn:compactMeasMatrixEq}
\boldsymbol{y_s} = \boldsymbol{C}\boldsymbol{q^a_c},
\end{equation}
where $\boldsymbol{q^a_c} =
\begin{pmatrix}
q^a_{x_1} & q^a_{x_2} & \dots q^a_{x_L}
\end{pmatrix}^T$ is a shortened version of $\boldsymbol{q^a}$ that only has $L$ dimensions. Also $\boldsymbol{C}$ is an $m {\times} L$ matrix of rank $L$, since $L{<}m$, and $\boldsymbol{h_{x_i}}$'s are linearly independent columns of $\boldsymbol{C}$ (recall our discussion in Section \ref{SufficientStatistic}).
Therefore, $\boldsymbol{C}$ has a left Moore-Penrose inverse (pseudo inverse), $\boldsymbol{C}^+ = (\boldsymbol{C}^T\boldsymbol{C})^{-1}\boldsymbol{C}^T$ where $\boldsymbol{C}^+\boldsymbol{C} = \boldsymbol{I}$ of size ${L{\times}L}$.
Thus, if we have knowledge of $\boldsymbol{C}$, we can then find $\boldsymbol{q^a_c}$ as:
\begin{equation}
\label{eqn:Eff_q_a_c}
\boldsymbol{q^a_c} = \boldsymbol{C}^+\boldsymbol{y_s}.
\end{equation}

The problem we need to solve is obtaining the matrix $\boldsymbol{C}$. We can solve this problem using an exhaustive search method which can be explained as follows:
\begin{enumerate}[(i)]
\item Candidate matrices $\boldsymbol{C_j}$ are generated by choosing different $L$ combinations of columns of $\boldsymbol{H}$ where $1 {\leq} j {\leq} {{n_r}\choose{L}}$.

\item Find $\boldsymbol{q^a_{c_j}} {=} \boldsymbol{C}^+_{\boldsymbol{j}}\boldsymbol{y_s} {=} \boldsymbol{C_j^{\text{$+$}}} \boldsymbol{C} \boldsymbol{q^a_c}$.
Note that at this step, we obtain a vector $\boldsymbol{q^a_{c_j}}$ identical to $\boldsymbol{q^a_c}$ if and only if $\boldsymbol{C}^+_{\boldsymbol{j}} \boldsymbol{C} {=} \boldsymbol{I} \Leftrightarrow \boldsymbol{C_j} {=} \boldsymbol{C}$.

\item Let $\boldsymbol{\beta}_j$ be such that 
\begin{equation}
\boldsymbol{\beta}_j = \boldsymbol{C_j} \boldsymbol{q^a_{c_j}} = \boldsymbol{C_j} \boldsymbol{C}^+_{\boldsymbol{j}} \boldsymbol{y_s},
\end{equation}
Hence, if the correct choice  $\boldsymbol{C_j}  = \boldsymbol{C}$ is made, then
\begin{align*}
\boldsymbol{\beta}_j &= \boldsymbol{C_j} \boldsymbol{C}^+_{\boldsymbol{j}} \boldsymbol{C} \boldsymbol{q^a_c}\\
&= \boldsymbol{C_j} \boldsymbol{q^a_c} = \boldsymbol{C} \boldsymbol{q^a_c}\\
&= \boldsymbol{y_s},
\end{align*}
else, if $\boldsymbol{C_j} \neq \boldsymbol{C}$, then\footnote{Since $\boldsymbol{C}^+_{\boldsymbol{j}}$ is the left pseudo-inverse of $\boldsymbol{C_j}$, and since $\boldsymbol{C_j}$ is not a square matrix, then
$\boldsymbol{C_j} \boldsymbol{C}^+_{\boldsymbol{j}} \neq \boldsymbol{I} = \boldsymbol{C}^+_{\boldsymbol{j}} \boldsymbol{C_j} \Longrightarrow \boldsymbol{C_j} \boldsymbol{C}^+_{\boldsymbol{j}} \boldsymbol{C} {\neq} \boldsymbol{C}$}
\begin{align*}
\boldsymbol{\beta}_j
&= \boldsymbol{C_j} \boldsymbol{C}^+_{\boldsymbol{j}} \boldsymbol{C} \boldsymbol{q^a_c}\\
& \neq \boldsymbol{y_s}
\end{align*}
%
\end{enumerate}

Hence, if $\boldsymbol{\beta}_{j^*} =  \boldsymbol{y_s}$, we declare its corresponding matrix $\boldsymbol{C}_{j^*}$ the true matrix $\boldsymbol{C}$ defined in Eq. (\ref{eqn:EffC}) which satisfies Eq. (\ref{eqn:compactMeasMatrixEq}).
Also, we have that $\boldsymbol{q^a_c} = \boldsymbol{q^a_{c_{j^*}}}$.
Since, identifying $\boldsymbol{C}$ is equivalent to identifying the indexes $x_1, \dots , x_L$. Thus, we found the angular channel $\boldsymbol{q^a}$ which is all zeros except - potentially\footnote{This means that if the number of paths is less than $L$, then some $q^a_{x_i}$'s might have zero values as well.} - for the components $q^a_{x_1}, \dots, q^a_{x_L}$.

The previous discussion dealt with an idealized version of the measurements (i.e., $\boldsymbol{y_s}$), however, in practice, we observe $\boldsymbol{u_s}$ as an error-corrupted version of $\boldsymbol{y_s}$.
Define $\boldsymbol{z_s}$ to be the error vector that captures the effect of both channel noise and quantization error which satisfies
\begin{equation}
\label{eqn:quantizedError}
\boldsymbol{u_s} = \boldsymbol{y_s} + \boldsymbol{z_s}.
\end{equation}

Suppose that we know the matrix $\boldsymbol{C}$ for which we have
\begin{equation}
\boldsymbol{u_s} = \boldsymbol{C} \boldsymbol{q^a_c} + \boldsymbol{z_s},
\end{equation}
then, we can find $\boldsymbol{C}^+ \boldsymbol{u_s}$ (compare to Eq. (\ref{eqn:Eff_q_a_c})) as follows
\begin{align*}
\boldsymbol{C}^+ \boldsymbol{u_s} &= \boldsymbol{C}^+ \boldsymbol{C} \boldsymbol{q^a_c} + \boldsymbol{C}^+ \boldsymbol{z_s} \\
&=  \boldsymbol{q^a_c} + \boldsymbol{C}^+ \boldsymbol{z_s},
\end{align*}
to be a noise-corrupted version of $\boldsymbol{q^a_c}$.

Now, to find an estimate $\hat{\boldsymbol{q}}^{\boldsymbol{a}}$ of $\boldsymbol{q^a}$, we follow a very similar procedure to the one described before as follows:
\begin{enumerate}[(i)]
\item Matrices $\boldsymbol{C_j}$ are generated similar to the $1^{st}$ step before.

\item Define $\boldsymbol{E_j}$ to be the difference between $\boldsymbol{C}$ and $\boldsymbol{C_j}$ where
\begin{equation}
\boldsymbol{C} = \boldsymbol{C_j} + \boldsymbol{E_j}.
\end{equation}
That is, $\boldsymbol{E_j} = \mathbf{0}$ (all zero matrix) $\Leftrightarrow \boldsymbol{C_j} = \boldsymbol{C}$.

\item Find $\boldsymbol{q^a_{c_j}}$ such that 
\begin{align}
\boldsymbol{q^a_{c_j}} &= \boldsymbol{C}^+_{\boldsymbol{j}}\boldsymbol{u_s} \\
&= \boldsymbol{C_j^{\text{$+$}}} \boldsymbol{C} \boldsymbol{q^a_c} + \boldsymbol{C_j^{\text{$+$}}} \boldsymbol{z_s} \nonumber \\
&= \boldsymbol{C_j^{\text{$+$}}} (\boldsymbol{C_j} + \boldsymbol{E_j}) \boldsymbol{q^a_c} + \boldsymbol{C_j^{\text{$+$}}} \boldsymbol{z_s}
\nonumber \\
&= \boldsymbol{q^a_c}  +  \boldsymbol{C_j^{\text{$+$}}}  (\boldsymbol{E_j} \boldsymbol{q^a_c} + \boldsymbol{z_s}).
\nonumber
\end{align}

Unlike the $2^{nd}$ step of the no-error case, $\boldsymbol{q^a_{c_j}}$ will not be identical to the true $\boldsymbol{q^a_c}$ with probability $1$, since $\boldsymbol{z_s}$ is not identical to $\mathbf{0}$ with probability $1$ ($\boldsymbol{z_s}$ is the difference between continuous and discrete quantities).

\item
Let $\boldsymbol{\beta}_j$ be such that
\begin{align}
\boldsymbol{\beta}_j &= \boldsymbol{C_j} \boldsymbol{q^a_{c_j}} = \boldsymbol{C_j}  \left( \boldsymbol{q^a_c}  +  \boldsymbol{C_j^{\text{$+$}}}  \left( \boldsymbol{E_j} \boldsymbol{q^a_c} + \boldsymbol{z_s}\right) \right) \\
&=
\begin{cases}
\boldsymbol{y_s} + \boldsymbol{C_j} \boldsymbol{C_j^{\text{$+$}}} \boldsymbol{z_s} &,  \boldsymbol{C} {=} \boldsymbol{C_j} \\
\boldsymbol{C_j} \boldsymbol{q^a_c} + \boldsymbol{C_j} \boldsymbol{C_j^{\text{$+$}}} \left( \boldsymbol{E_j} \boldsymbol{q^a_c} + \boldsymbol{z_s} \right)
&, \boldsymbol{C} {\neq} \boldsymbol{C_j}
\end{cases}
\end{align}
Then find $j^*$ such that 
\begin{equation}
j^* = \argmin_j \norm{\boldsymbol{\beta}_j - \boldsymbol{u_s}} ,
\end{equation}
such that $\boldsymbol{\beta}_j - \boldsymbol{u_s}$ is given by
\begin{equation}
\boldsymbol{\beta}_j {-} \boldsymbol{u_s} =
\boldsymbol{C_j} \boldsymbol{q^a_c} {-} \boldsymbol{y_s}
+
\boldsymbol{C_j} \boldsymbol{C_j^{\text{$+$}}} \boldsymbol{E_j} \boldsymbol{q^a_c} {+}
\left(\boldsymbol{C_j} \boldsymbol{C_j^{\text{$+$}}} {-} \boldsymbol{I}\right)\boldsymbol{z_s}
\end{equation}
which at $\boldsymbol{C} = \boldsymbol{C_{j^*}}$ is further reduced to
\begin{equation}
\boldsymbol{\beta}_{j^*} - \boldsymbol{u_s} =
\left(\boldsymbol{C_j} \boldsymbol{C_j^{\text{$+$}}} {-} \boldsymbol{I}\right)\boldsymbol{z_s}
\end{equation}
\end{enumerate}

\subsection{Multiple Transmit and Receive Antennas}
\label{MultiTXRX}
So far, we have considered channels with single transmit antennas and shown how to perform beam discovery at RX. To extend our approach to a general setting, we consider channels with $n_t$ antennas at TX, and $n_r$ antennas at RX.
Thus, instead of the TX just sending signals omnidirectionally, now it can perform highly directional transmission.
Recall that the RX is able to perform channel measurements using multi-armed beams. Similarly, the TX can send signals using multi-armed beams to simultaneously focus on multiple directions using precoding vectors $\boldsymbol{f_j}$.

The design of precoding vectors can also be obtained using an LBC approach. Similar to the method of designing rx-combining vectors $\boldsymbol{w_i}$, we look for an LBC, $C_2$, that has CWs of length $n_2 {=} n_t$ and can correct for $e_n {=} L$ errors. Let the parity check matrix of $C_2$ be $\boldsymbol{H}_2$, using which, we will design the precoding vectors $\boldsymbol{f_j}$.
Let \textit{beam}$_{\text{\textit{tx}}}$\textit{\#i} denote the $i^{th}$ TX beam which points to TX direction \textit{dir}$_{\text{\textit{tx}}}$ \textit{\#i}.
Then, just as before, we envisage $\boldsymbol{H}_2$ as an array whose columns are associated with resolvable TX directions such that: i) its $j^{th}$ column corresponds to \textit{dir}$_{\text{\textit{tx}}}$\textit{\#j}, and ii) its $i^{th}$ row corresponds to the $i^{th}$ measurement.
We note that no actual measurements are performed at TX; we use the word \textit{measurement} to refer to precoding, consistent with the case of RX.
That is, the $i^{th}$ TX measurement is actually the $i^{th}$ precoder $\boldsymbol{f}_i$.
Thereby, we design the $i^{th}$ precoder as a multi-armed TX beam such that, only if $h_{i,j}$, the intersection of the $i^{th}$ row and $j^{th}$ columns of $\boldsymbol{H}_2$, is $ = 1$, do we include sub-beam \textit{beam}$_{\text{\textit{tx}}}$\textit{\#j} in $\boldsymbol{f_i}$.
Each TX measurement provides a component in a TX channel syndrome vector $\boldsymbol{y_s^{TX}}$. The total number of TX measurements (i.e., precoding vectors), denoted by $m_2$, is equal to the number of parity check bits of the code $C_2$. That is, $m_2 {=} n_2{-}k_2$, where $k_2$ is the length of $C_2$'s information sequences.
To obtain AoDs of strong paths at TX, we define the function $\xi_2()$ as the mapping function between all possible TX channel syndromes and their corresponding angular channels denoted by ${\boldsymbol{q}^a}^{\boldsymbol{TX}}$.
Note that, for every \textit{dir$_{\text{\textit{rx}}}$\#i}, there exists a corresponding ${\boldsymbol{q}^a}^{\boldsymbol{TX(i)}}$ which represents the $i^{th}$ row of $\boldsymbol{Q^a}$.
Also, since the maximum number of clusters is $L$, then, the number of non-zero vectors ${\boldsymbol{q}^a}^{\boldsymbol{TX(i)}}$ is $\leq L$.

\IncMargin{1em}
\begin{algorithm}[t]
\SetKwData{Left}{left}\SetKwData{This}{this}\SetKwData{Up}{up}
\SetKwFunction{Union}{Union}\SetKwFunction{FindCompress}{FindCompress}
\SetKwInOut{Input}{input}\SetKwInOut{Output}{output}

 \Input{$ \{ \boldsymbol{w}_i\}_{ \forall i \in \{1,\dots,m_1\}} \:$,
 $\{ \boldsymbol{f}_j \} _{, \forall j \in \{1,\dots,m_2\}} \:$ ,
 $\xi_1()  :  \boldsymbol{y_{s}}  \rightarrow \widehat{\boldsymbol{q}}^a \:$,
 $\xi_2()  :  \boldsymbol{y_{s}^{TX}}  \rightarrow \hat{\boldsymbol{q}}^{a^{\boldsymbol{TX}}}$
 }
 \Output{ $\{\boldsymbol{y_{s_i}}\}_{\forall i \in \{1,\dots,m_2\}}$ }
\Begin{
 $j = 0$\;
 \While{$j < m_2 $}{
 	$i = 0$\;
	\While{$i < m_1 $}{
		$y_{s_{i,j}} = \boldsymbol{w}_i^H \boldsymbol{Q} \boldsymbol{f}_j s + \boldsymbol{w}_i^H \boldsymbol{n}$  \tcp*[r]{channel measurement}
		$i \leftarrow i + 1$
	}
	$\boldsymbol{y_{s_j}} \leftarrow \{y_{s_{i,j}}\}_{\forall i \in \{1,\dots,m_1\}}$ \tcp*[r]{construct channel syndrome $\boldsymbol{y_{s_j}}$}
		\tcc{find corresponding channel $\boldsymbol{q^a}^{(j)} = [{q_1^a}^{(j)}, {q_2^a}^{(j)}, \dots, {q_{n_r}^a}^{(j)}]^T$}
	$\boldsymbol{q^a}^{(j)} \leftarrow \xi_1(\boldsymbol{y_{s_j}})$\;

		\For{$p\leftarrow 1$ \KwTo $n_r$}{

			\tcc{construct TX channel syndromes $\boldsymbol{y_{s}^{TX(p)}}$, where $\boldsymbol{y_{s}^{TX(p)}} = [y_{s_1}^{TX(p)},y_{s_2}^{TX(p)}, \dots, y_{s_{m_2}}^{TX(p)}]^T$}		
			$y_{s_j}^{TX(p)} \leftarrow {q_p^a}^{(j)}$

			}

		$j \leftarrow j + 1$\;
 }
	\For{$p\leftarrow 1$ \KwTo $n_r$}{
 $\boldsymbol{{q^a}^{TX(p)}} \leftarrow \xi_2(\boldsymbol{y_s^{TX(p)}}) $
 }
 $\widehat{\boldsymbol{Q}}^{\boldsymbol{a}} =
 \begin{pmatrix}
  \boldsymbol{{q^a}^{TX(1)}} & \boldsymbol{{q^a}^{TX(2)}} & \hdots & \boldsymbol{{q^a}^{TX(n_r)}}
 \end{pmatrix}^T
 $
}
 \caption{\small Beam discovery of multiple TX/RX antennas.}
 \label{Alg:MultiTXRX}
\end{algorithm}
\DecMargin{1em}

To see the whole picture, assume that a code $C_1$, with CWs of length $n_1 {=} n_r$, is an LBC code associated with beam discovery at RX side. Let the number of RX measurements, i.e., the number of rx-combining vectors, be $m_1$ such that $m_1  {=} n_1 {-} k_1$, where $k_1$ is the length of information sequences of $C_1$. Also let $\xi_1()$ be the mapping function between RX channel syndromes and its corresponding angular channel.
Under this setting, the beam discovery problem is performed as follows: 
i) The TX starts starts sending its training sequence using the precoder $\boldsymbol{f_j}, \forall j \in \{0,\dots,m_2{-}1\}$.
ii) The RX performs $m_1$ channel measurements while $\boldsymbol{f_j}$ is being used at TX and obtains a channel syndrome $\boldsymbol{y_{s_j}}$.
iii) Based on $\boldsymbol{y_{s_j}}$, the RX obtains a corresponding channel, $\boldsymbol{{q^a}^{(j)}}$ with path components $\{{q_p^a}^{(j)}\}_{\forall p \in \{1,\dots,n_r\}}$.
Notice that $\boldsymbol{{q^a}^{(j)}}$'s do not necessarily represent individual path gains, but rather, combinations of paths accumulating at a single \textit{dir$_{\text{\textit{rx}}}$\textit{\#}}. Therefore, there exists a resemblance to channel syndromes which we exploit.
iv) We construct a set of $n_r$ TX channel syndromes, $\boldsymbol{y_s^{TX(p)}}$ where their $j^{th}$ component $y_{s_j}^{TX(p)} {=} {q_p^a}^{(j)}$, i.e.,
$[\boldsymbol{y_s^{TX(1)}}, \boldsymbol{y_s^{TX(2)}}, \dots, \boldsymbol{y_s^{TX(n_r)}}] =
[\boldsymbol{{q^a}^{(1)}}, \boldsymbol{{q^a}^{(2)}} , \dots, \boldsymbol{{q^a}^{(m_2)}}]^T$.
v) Finally, we find the angular TX channel for every \textit{dir}$_\text{\textit{rx}}$ \textit{\#p}, i.e., $p^{th}$ row of $\boldsymbol{Q^a}$, using the mapping function $\boldsymbol{{q^a}^{TX(p)}} {=} \xi_2(\boldsymbol{y_s^{TX(p)}})$.
Notice that, since no more than $L \ll n_r$ clusters exist, and since $\boldsymbol{0}$ channels correspond to $\boldsymbol{0}$ channel syndromes, we only need to apply $\xi_2()$ a maximum of $L$ times -unless measurement error occurs.
This whole process is highlighted in Algorithm \ref{Alg:MultiTXRX}.

\begin{remark}{}
The estimated channel $\widehat{\boldsymbol{Q}}^{\boldsymbol{a}}$ may contain more than $L$ non-zero components. The reason is that the receiver obtains a channel $\boldsymbol{{q^a}^{(j)}}$ for every precoder $\boldsymbol{f_j}$ which may contain erroneous component estimates. Incorrect estimates occur as a result of measurement errors which happen due to i) channel noise, ii) quantization error.
Now every $\boldsymbol{{q^a}^{(j)}}$ may contain a maximum of $L$ non-zero components, however, some of which may be due to measurement errors. Afterwards, potentially noise-corrupted $\{\boldsymbol{{q^a}^{(j)}}\}_{\forall j\in\{1,\dots,n_r\}}$ are used to obtain TX channel syndromes as shown in Algorithm \ref{Alg:MultiTXRX}. That is, for every \textit{dir}$_{\text{\textit{rx}}}$\textit{\#i} we obtain a TX channel syndrome to identify the corresponding \textit{dir}$_{\text{\textit{tx}}}$\textit{\#j}'s that have strong components. Thus, we may obtain a maximum of $L$ non-zero components per \textit{dir}$_{\text{\textit{rx}}}$\textit{\#i}.
\end{remark}

\section{Error Correction}
\label{ErrorCorr}
So far, the main focus of our work has been finding the most efficient way for beam discovery under the channel sparsity assumption. While doing so, we did not really have special treatment to deal with measurement errors.
In fact, with no measurement errors, our proposed solution can estimate the channel matrix perfectly.
However, the presence of channel noise and quantization degrades the beam discovery (channel estimation) performance.
To combat the effect of such imperfections, we focus our attention on answering the following question: \textbf{\textit{Can we trade efficiency for performance?}}
By \textit{efficiency} we refer to the reduced number of measurements. So, in other words, we need to study whether increasing the number of channel measurements --- by essentially adding redundancy --- would improve the beam detection performance.
The answer to this is: \textbf{Yes}.
In the sequel we will present a method that allows for increasing the number of measurements and trades it for higher reliability.

The very concept of adding redundant information to combat noisy observations is the foundation of channel coding. Hence, it is appealing to use channel coding ideas to achieve more reliable beam discovery.
For simplicity we again present our proposed solution for the simple setting of one transmit antenna and multiple receive antennas. The general multiple transmit and receive antennas setting can be dealt with in the same fashion described in Section \ref{MultiTXRX}.

Recall that a received symbol $u_s$ is given by Eq. (\ref{eqn:quantizedRxSymbol}) as $u_s {=} [y_s + \boldsymbol{w}^H \boldsymbol{n}]_+$ where $y_s {=} \boldsymbol{w}^H \boldsymbol{Q} \boldsymbol{f} s$ is the error-free measurement symbol. We write $u_{s} {=} y_s + z_s$ where $z_s$ is the measurement error (Eq. (\ref{eqn:quantizedError})).
Also recall that, for $\boldsymbol{f} {=} 1$ (one transmit antenna), and $\boldsymbol{w_i}$, we form the channel syndrome vector $\boldsymbol{u_s} {=} [u_{s_0} \: u_{s_1} \: \dots \: u_{s_{m{-}1}}]^T$ such that $u_{s_i} {=}  \boldsymbol{w_i}^H \boldsymbol{q} s {+} z_{s_i}$ where $\boldsymbol{q}$ is the $n_r {\times} 1$ channel vector. Equivalently, we have that $\boldsymbol{u_s} {=} \boldsymbol{H} \boldsymbol{q^a} {+} \boldsymbol{z_s}$ where $\boldsymbol{z_s}$ is formed by stacking $\{z_{s_i}\}_{\forall i {=} 0,\dots,m{-}1}$. Recall that this is exactly Eq. (\ref{syndromeEqn}) but with the noise terms added.

In fact, we can perceive the channel syndrome $\boldsymbol{y_s}$ as raw information sequence that need to be transmitted over a noisy channel, and $\boldsymbol{u_s}$ is the noise-corrupted received sequence. The syndrome, $\boldsymbol{y_s}$, is a vector that lies in an $m-$dimensional vector space. By exploiting channel codes, we can map $\boldsymbol{y_s}$ to longer sequences $\boldsymbol{y_s^{\nu}}$ (encoded channel syndrome) that lie in an $m-$dimensional subspace of an $m_c-$dimensional vector space.
The longer sequences $\boldsymbol{y_s^{\nu}}$ should have increased distance which allows for higher resilience against measurement errors
Hence, $\boldsymbol{u_s}$ can now be written as $\boldsymbol{u_s} = \boldsymbol{y_s^{\nu}} + \boldsymbol{z_s}$. Our goal is to design $\boldsymbol{y_s^{\nu}}$. Once we achieve that, the rest of the problem can be tackled as discussed in section \ref{BeamDiscovery}.

Towards that end, let us use an error correction code $C_c$, with generator matrix $\boldsymbol{G_c}$ and error correction capability $e_c$. Note that we use the subscript $c$ to refer to \textit{correction}. The size of $\boldsymbol{G_c}$ is $m{\times}m_c$.
Thus, the encoded channel syndromes can be represented as $\boldsymbol{y_s^{\nu}} = \boldsymbol{G}^T_{\boldsymbol{c}} \boldsymbol{y_s}$,
where $\boldsymbol{y_s}$ and $\boldsymbol{y_s^{\nu}}$ are of sizes $m{\times}1$ and $m_c{\times}1$, respectively.
Thus, $\boldsymbol{y_s^{\nu}}$ can be written as $\boldsymbol{y_s^{\nu}} = \boldsymbol{G}^T_{\boldsymbol{c}} \boldsymbol{H} \boldsymbol{q^a}$. Then, similar to Eqn. (\ref{syndromeEqn}) we want to use the matrix $\boldsymbol{G}^T_{\boldsymbol{c}} \boldsymbol{H}$ to design $\boldsymbol{y_s^{\nu}}$. However, the problem here is that this matrix in not necessarily a binary matrix (i.e., with elements of $'1'$s and $'0'$s). Hence, let us denote by $\boldsymbol{H^{\nu}}$, the matrix $\boldsymbol{G}^T_{\boldsymbol{c}} \boldsymbol{H} \pmod {2}$ and use it to design $\boldsymbol{y_s^{\nu}}$ such that
\begin{equation}
\label{eqn:encodedChSyn}
\boldsymbol{y_s^{\nu}} =  \boldsymbol{H^{\nu}}\boldsymbol{q^a}.
\end{equation}
In other words, $\boldsymbol{H^{\nu}} {=} \boldsymbol{G}^T_{\boldsymbol{c}} \boldsymbol{H}$
is the matrix product over $GF(2)$.
Therefore, instead of designing the channel measurements based on $\boldsymbol{H}$, we propose to design them based on $\boldsymbol{H^{\nu}}$ with that being the only difference to the design proposed earlier.

At this point, it remains to show that the new measurements design still provides a one-to-one mapping to every angular channel (i.e., if $\boldsymbol{q^a_1} \neq \boldsymbol{q^a_2}$, then their corresponding channel syndromes $\boldsymbol{y_{s1}^{\nu}} \neq \boldsymbol{y_{s2}^{\nu}}$). Furthermore, we will show that the new design provides a better resilience to measurement errors.
That is, we will show that if $\boldsymbol{q^a_1} \neq \boldsymbol{q^a_2}$, then $\delta(\boldsymbol{y_{s1}},\boldsymbol{y_{s2}}) \leq \delta(\boldsymbol{y_{s1}^{\nu}},\boldsymbol{y_{s2}^{\nu}})$, where $\delta(\cdot,\cdot)$ is defined as in Eq. (\ref{eqn:distance}).

\subsection{Sufficient Statistic}

We start off by showing that the new measurements provide a sufficient statistic for beam discovery.
We will follow a similar approach to that of Section \ref{SufficientStatistic}. Specifically, we will first consider error patterns, matrices, and operators over the finite field $GF(2)$.
Afterwards, we will extend those concepts to the complex field where all channel matrices and measurements lie.

Let us consider a code $C$, with codewords of length $n$ and error correction capability $e_n$.
The parity check and generator matrices of $C$ are given by $\boldsymbol{H}$ and $\boldsymbol{G}$, respectively.
The error syndromes of $C$ are given by $\boldsymbol{s} = \boldsymbol{r}\boldsymbol{H}^T = \boldsymbol{e}\boldsymbol{H}^T$, where $\boldsymbol{r}$ is the received sequence and $e$ is the error pattern corrupting the transmitted codeword $\boldsymbol{c}$ (recall Footnote \ref{LBC}).
Now suppose we encode $\boldsymbol{s}$ using another error correction code $C_c$. The parity check and generator matrices of $C_c$ are given by $\boldsymbol{H_c}$ and $\boldsymbol{G_c}$, respectively.
The encoded syndromes $\boldsymbol{s^{\nu}}$ are given as
\begin{equation}
\label{eqn:encodedSyndrome}
\boldsymbol{s^{\nu}} 
\stackrel{\smash{\scriptscriptstyle (a)}}{=}
\boldsymbol{s} \boldsymbol{G_c}
\stackrel{\smash{\scriptscriptstyle (b)}}{=}
\boldsymbol{e}\boldsymbol{H}^T \boldsymbol{G_c}
\stackrel{\smash{\scriptscriptstyle (c)}}{=}
\boldsymbol{e} \boldsymbol{H^{\nu}}^T,
\end{equation}

Consider all single bit error patterns $\boldsymbol{e}^{(i)}$ as defined in Eq. (\ref{eqn:singleBitErrors}).
Let the encoded syndrome that corresponds to $\boldsymbol{e}^{(i)}$ be $\boldsymbol{s^{\nu}}^{(i)} = \boldsymbol{e}^{(i)} \boldsymbol{H^{\nu}}^T$.
Thus, $\boldsymbol{s^{\nu}}^{(i)}$ is exactly the $i^{th}$ row of $\boldsymbol{H^{\nu}}^T$, i.e., $i^{th}$ column of $\boldsymbol{H^{\nu}}$.
%
%
%
\begin{lemma}\label{lemma:Uni_ErrCrr}
For any error sequence $\boldsymbol{e_t}$ with number of bit errors identical to $e_n$, its encoded syndrome $\boldsymbol{s_t}^{\boldsymbol{\nu}}$ is a linear combination of $e_n$ linearly independent vectors $\boldsymbol{s^{\nu}}^{(i)}$.
\end{lemma}
See Appendix \ref{append:proof_LemmaUniErrCrr} for proof of Lemma \ref{lemma:Uni_ErrCrr}.

Lemma \ref{lemma:Uni_ErrCrr} allows us to use the result of Lemma \ref{lemma:LinInd} which states that if we have a collection, $\big\{\boldsymbol{s^{\nu}}^{(i)}\big\}_{i{=}x_1}^{x_{e_n}}$, of linearly independent vectors over $GF(2)$. Then, if their $'1'$ and $'0'$ entries are interpreted as real numbers, then they are also linearly independent over $\mathbb{C}$.

Let us interpret the elements of $\boldsymbol{H^{\nu}}$ and $\boldsymbol{e^{(i)}}$ as real numbers. Then, we can write the channel $\boldsymbol{q^a}$ as
\begin{equation}
({\boldsymbol{q}^a})^T = \alpha_1 \boldsymbol{e}^{(1)} + \alpha_2 \boldsymbol{e}^{(2)} + \dots + \alpha_n \boldsymbol{e}^{(n)}
\end{equation}
where $\alpha_i {\in} \mathbb{C}$
and $\sum_{i=1}^n \mathbb{\mathbbm{1}}_{\left\lbrace\alpha_i {\neq} 0\right\rbrace} {\leq} e_n$.
Therefore, each encoded channel syndrome ${(\boldsymbol{y_s^{\nu}})}^T {=} {(\boldsymbol{q}^a)}^T \boldsymbol{H^{\nu}}^T \Longrightarrow \boldsymbol{y_s^{\nu}} {=} \boldsymbol{H^{\nu}} \boldsymbol{q}^a$ is a linear combination of independent vectors in $\mathbb{C}^{m_c}$ (columns of $\boldsymbol{H^{\nu}}$).
Therefore, for all measurable channels $\boldsymbol{q^a_1} {\neq} \boldsymbol{q^a_2}$, we have that $\boldsymbol{y_{s1}^{\nu}} {\neq} \boldsymbol{y_{s2}^{\nu}}$. Therefore, measurements designed based on $\boldsymbol{H^{\nu}}$ are sufficient for beam discovery.

\subsection{Resilience to Errors}

In the next discussion, we are going to show that the encoded syndromes $\boldsymbol{y_s^{\nu}}$ are more tolerant to the occurrence of measurement errors.
Since the mapping functions $\xi()$ finds the correct $\boldsymbol{q^a}$ by using $l^2-$norm minimization methods (this is true for both look-up table and search methods), then it is intuitively beneficial to separate the channel syndrome vectors, in the $l^2-$norm sense, as much as possible.
Thus, we want to show that if two channel syndromes $\boldsymbol{y_{s_1}}$ and $\boldsymbol{y_{s_2}}$ (corresponding to channel vectors $\boldsymbol{q^a_1}$ and $\boldsymbol{q^a_1}$) have distance
$\delta(\boldsymbol{y_{s_1}}, \boldsymbol{y_{s_2}})$,
then their corresponding encoded syndromes are such that
\begin{align}
\delta(\boldsymbol{y_{s_1}^{\nu}}, \boldsymbol{y_{s_2}^{\nu}}) &\geq
\delta(\boldsymbol{y_{s_1}}, \boldsymbol{y_{s_2}}) \\
\Longleftrightarrow \quad
\norm{\boldsymbol{y_{s_1}^{\nu}} - \boldsymbol{y_{s_2}^{\nu}}}  &\geq
\norm{\boldsymbol{y_{s_1}}       - \boldsymbol{y_{s_2}}      }
\end{align}

\begin{proposition}\label{prop:Distance_ErrCrr}
Let $\boldsymbol{G_c}$ be the generator matrix of some LBC code $C$.
Then, if $\boldsymbol{G_c}$ is represented in the standard form and $\boldsymbol{H^{\nu}}$ is generated using $\boldsymbol{G_c}$ (as $\boldsymbol{G}^T_{\boldsymbol{c}} \boldsymbol{H} \pmod {2}$), then
$\norm{\boldsymbol{y_{s_1}^{\nu}} - \boldsymbol{y_{s_2}^{\nu}}}  \geq
\norm{\boldsymbol{y_{s_1}}       - \boldsymbol{y_{s_2}}      }$.
\end{proposition}

Proposition \ref{prop:Distance_ErrCrr} shows that by using any appropriate systematic code, we obtain encoded channel syndromes, $\boldsymbol{y_s^\nu}$, that have greater $l^2-$distance than the original syndromes $\boldsymbol{y_s}$. Hence, we increase the space allowed for the noise-corrupted measurement vector $\boldsymbol{u_s}$ to lie in, while still being able to identify its true corresponding, error-free, channel syndrome.
The proof of Proposition \ref{prop:Distance_ErrCrr} is provided in Appendix \ref{append:proof_Prop}.

\section{Performance Evaluation}
\label{Eval}
\begin{figure*}[t]
\centering
\begin{subfigure}{.33\textwidth}\label{fig:beamDProb}
  \centering
  \captionsetup{justification=centering}
\includegraphics[width=0.95\linewidth]{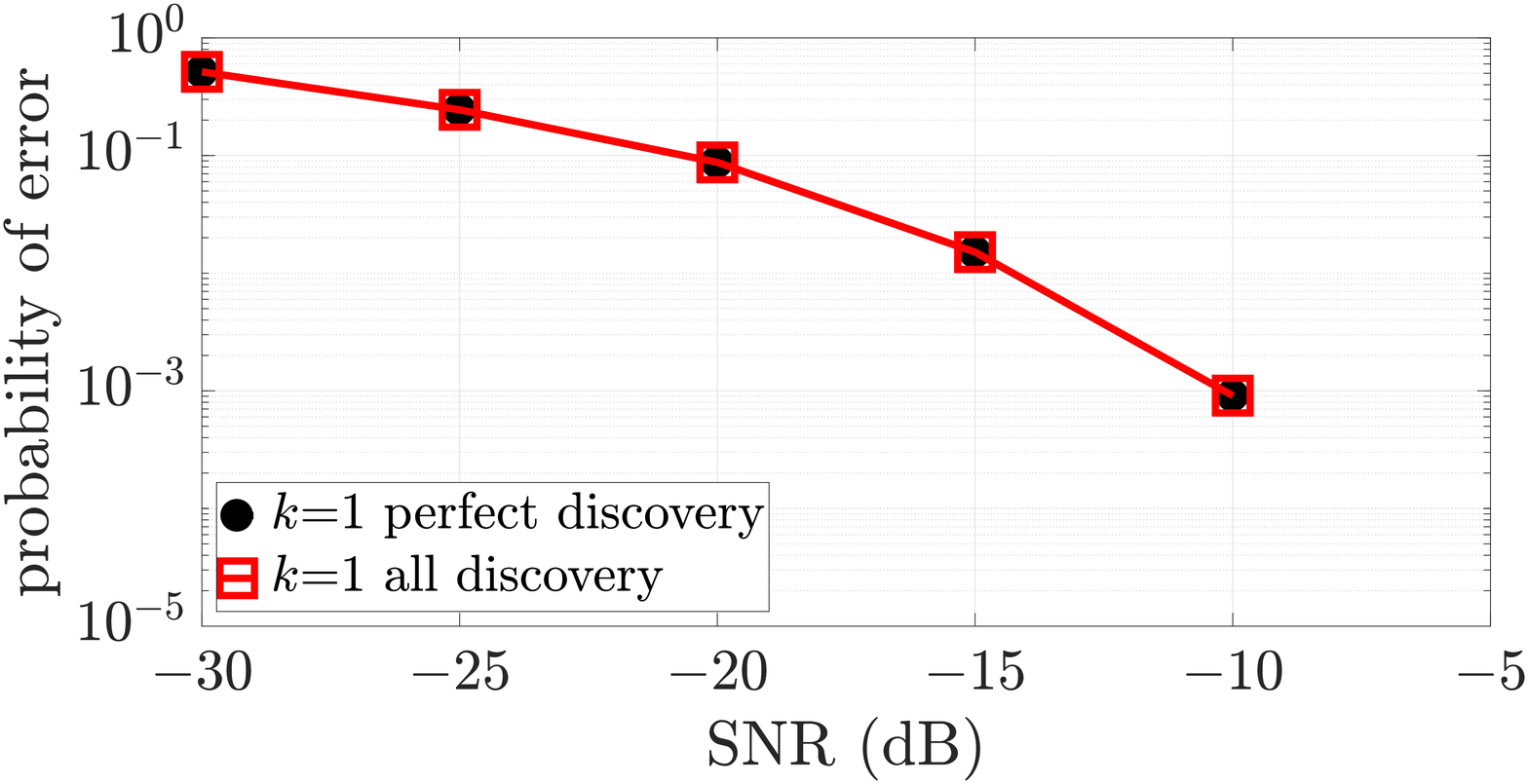}
\captionsetup{justification=centering}
\caption{\small $15 {\times} 15$ channel with $L{=}1$.}
\label{fig:beamDetProb_15x15_9L_1P}
\end{subfigure}%
\begin{subfigure}{.33\textwidth}
  \centering
  \captionsetup{justification=centering}
\includegraphics[width=0.95\linewidth]{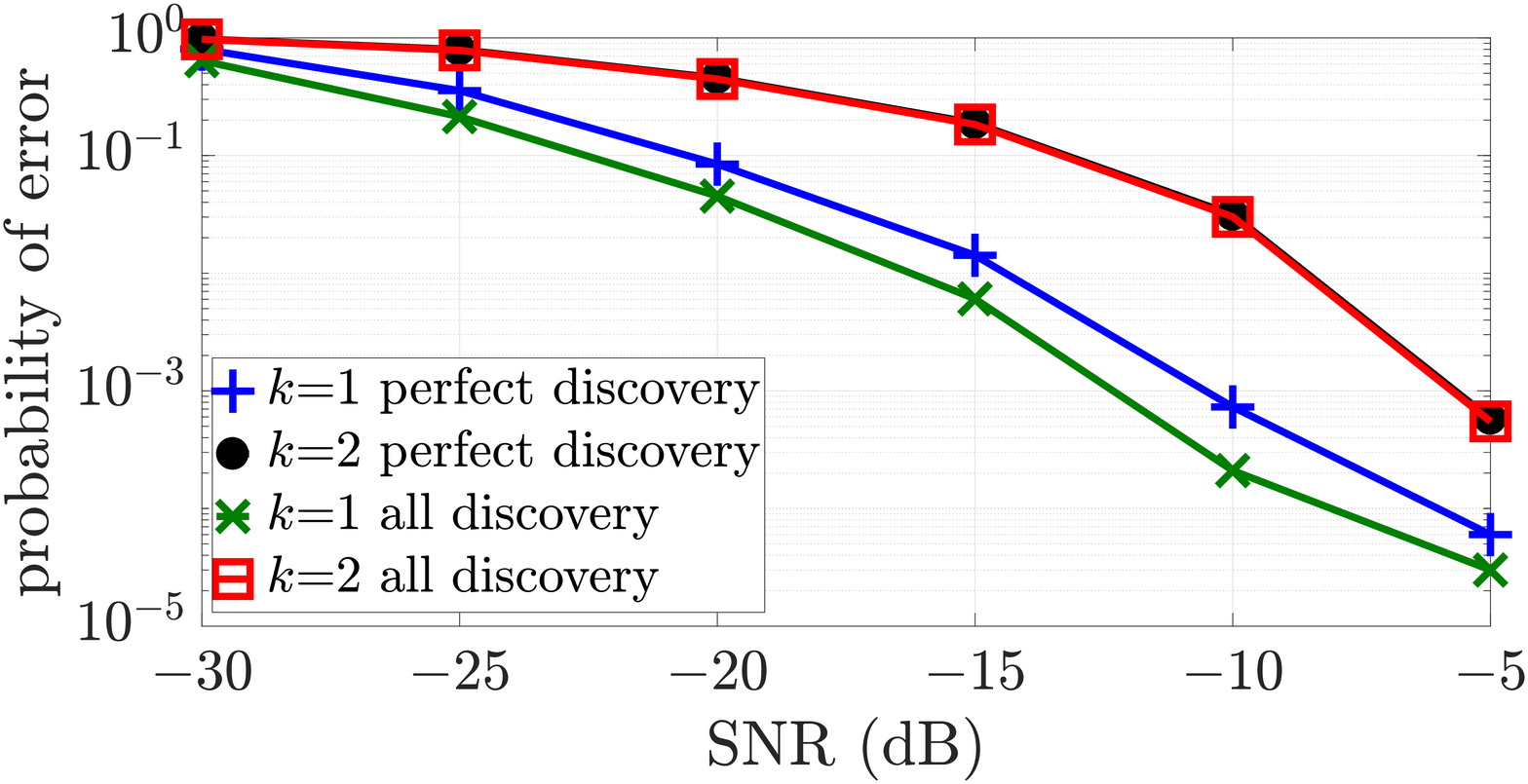}
\captionsetup{justification=centering}
\caption{\small $8 {\times} 8$ channel with $L{=}2$.}
\label{fig:beamDetProb_8x8_9L_2P}
\end{subfigure}%
\begin{subfigure}{.33\textwidth}
  \centering
\includegraphics[width=0.95\linewidth]{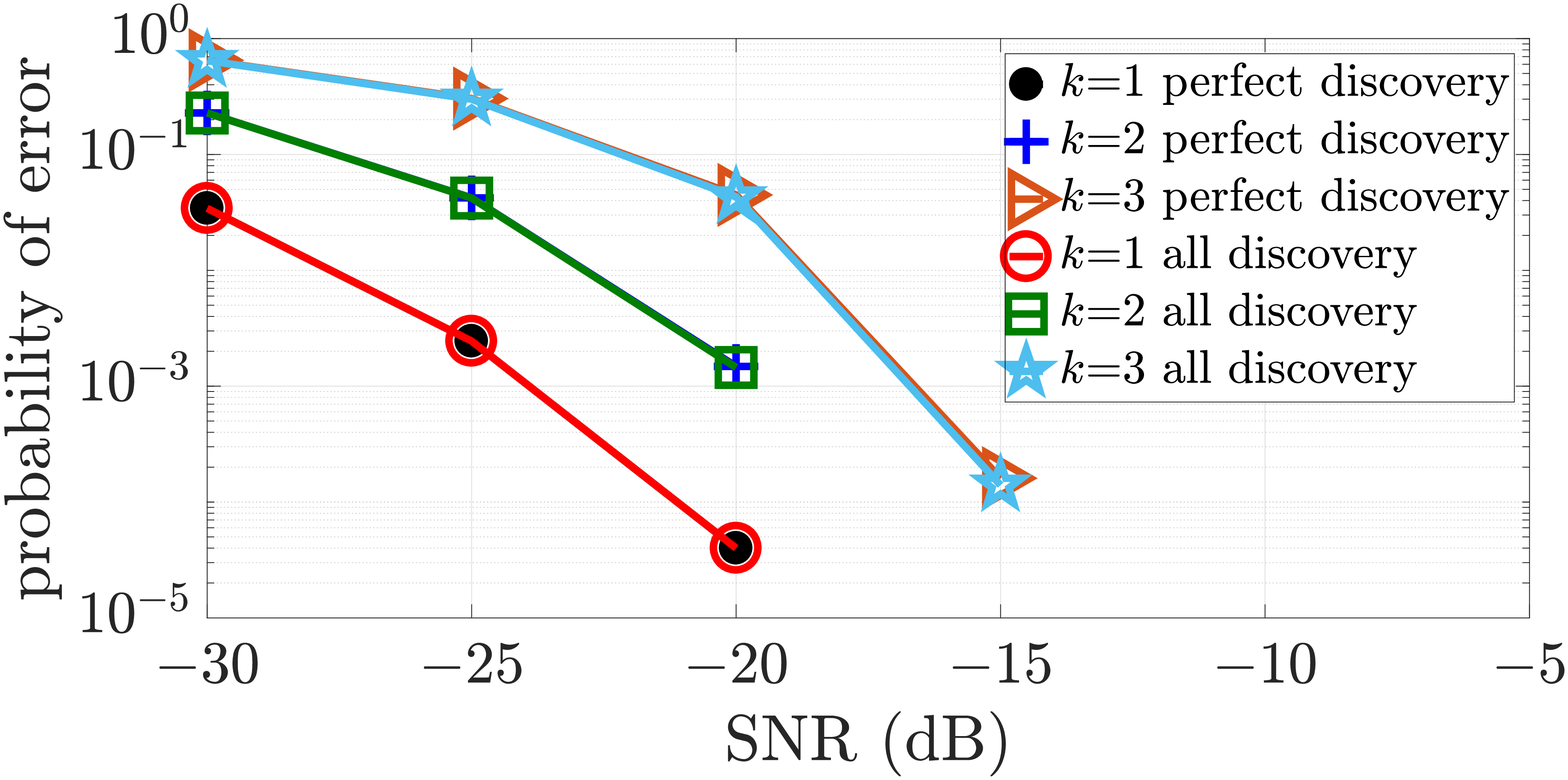}
\captionsetup{justification=centering}
\caption{\small $32 {\times} 32$ channel with $L{=}3$.}
\label{fig:beamDetProb_32x32_9L_3P}
\end{subfigure}%
\captionsetup{justification=centering}
\caption{\small Beam detection probability.}
\end{figure*}
\begin{figure*}[t]
\centering
\begin{subfigure}{.33\textwidth}\label{fig:MSE}
  \centering
  \captionsetup{justification=centering}
\includegraphics[width=0.95\linewidth]{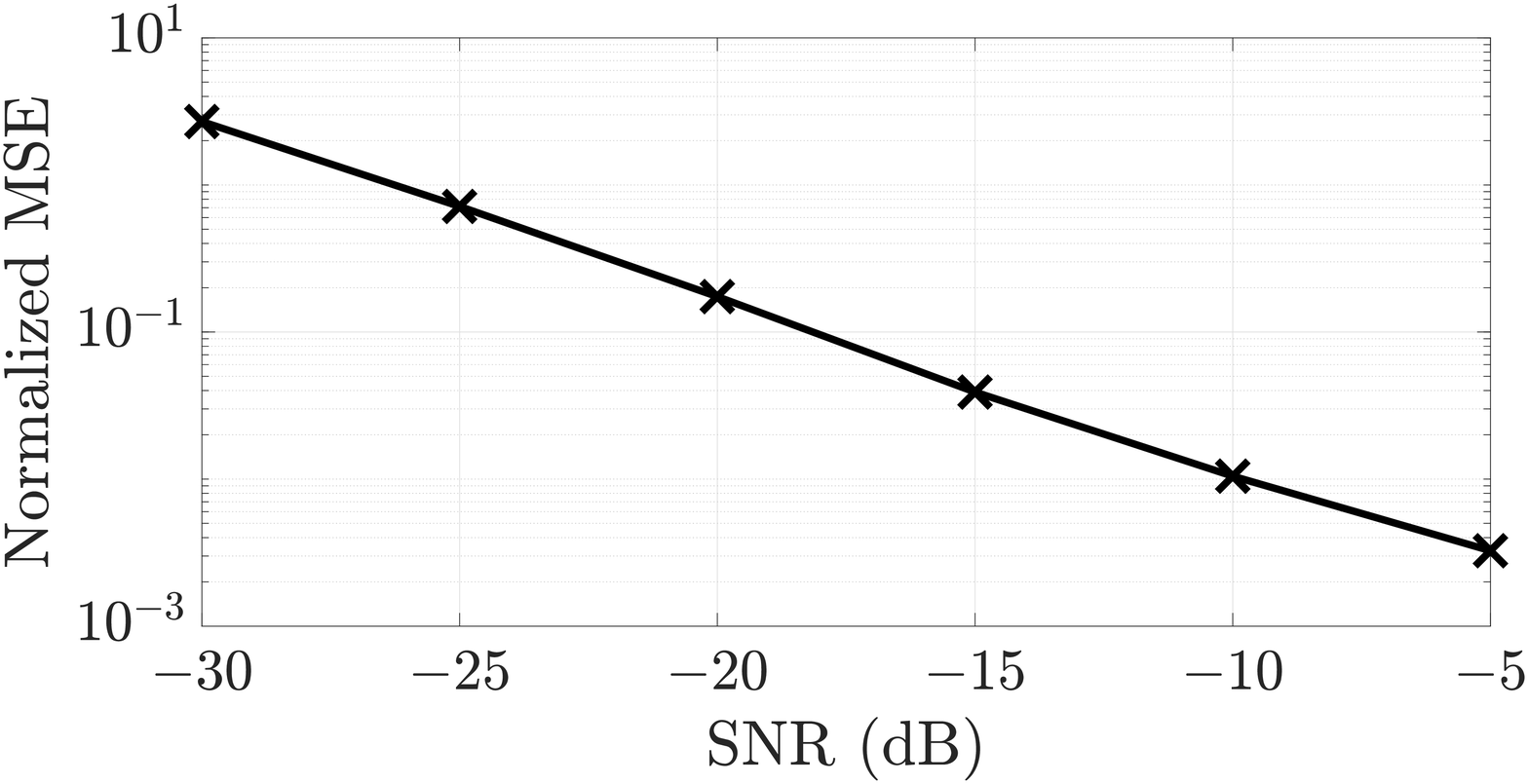}
\caption{\small $15 {\times} 15$ channel with $L{=}1$.}
\label{fig:MSE_15x15_9L_1P}
\end{subfigure}%
\begin{subfigure}{.33\textwidth}
  \centering
  \captionsetup{justification=centering}
\includegraphics[width=0.95\linewidth]{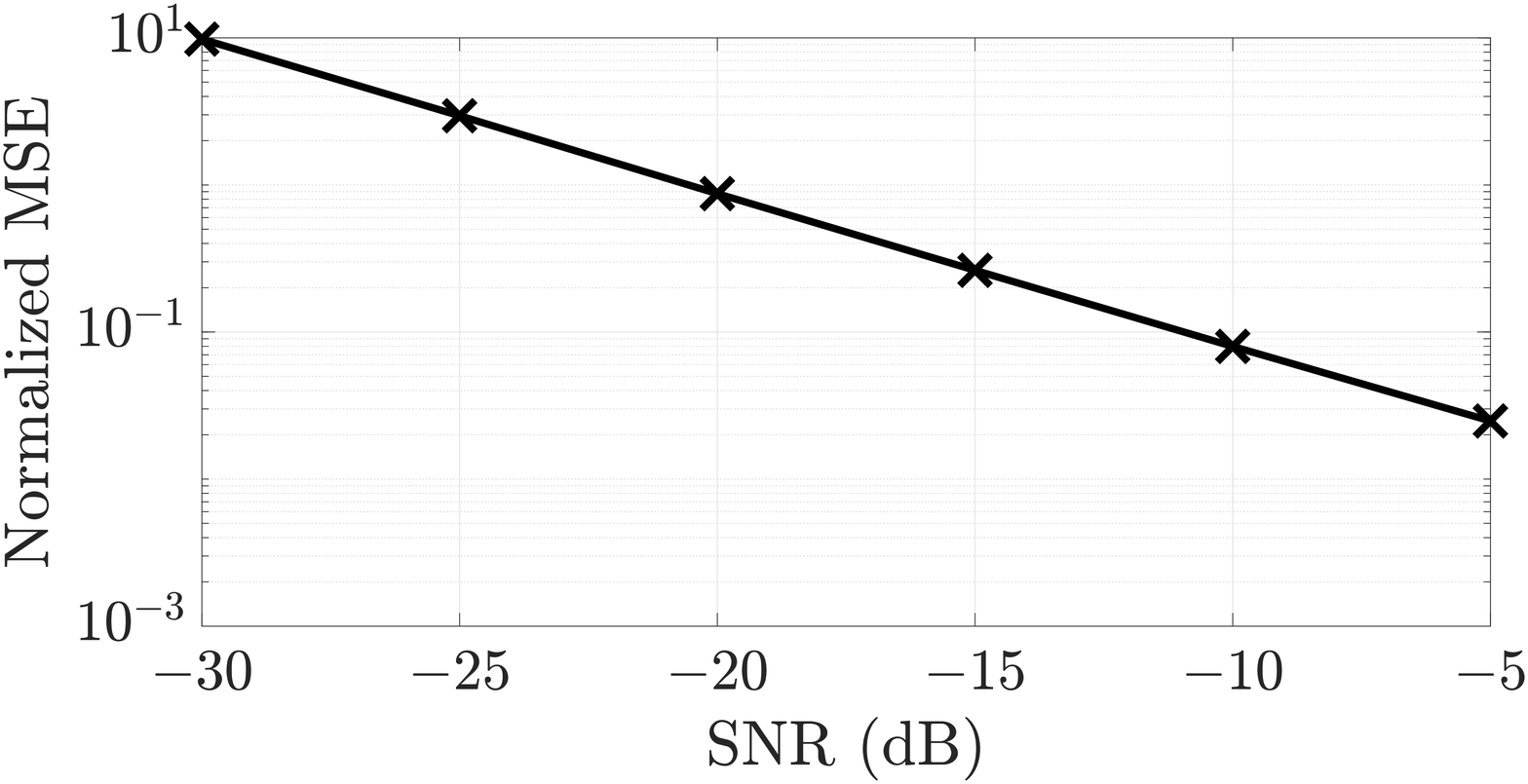}
\caption{\small $8 {\times} 8$ channel with $L{=}2$.}
\label{fig:MSE_8x8_9L_2P} 
\end{subfigure}%
\begin{subfigure}{.33\textwidth}
  \centering
\includegraphics[width=0.95\linewidth]{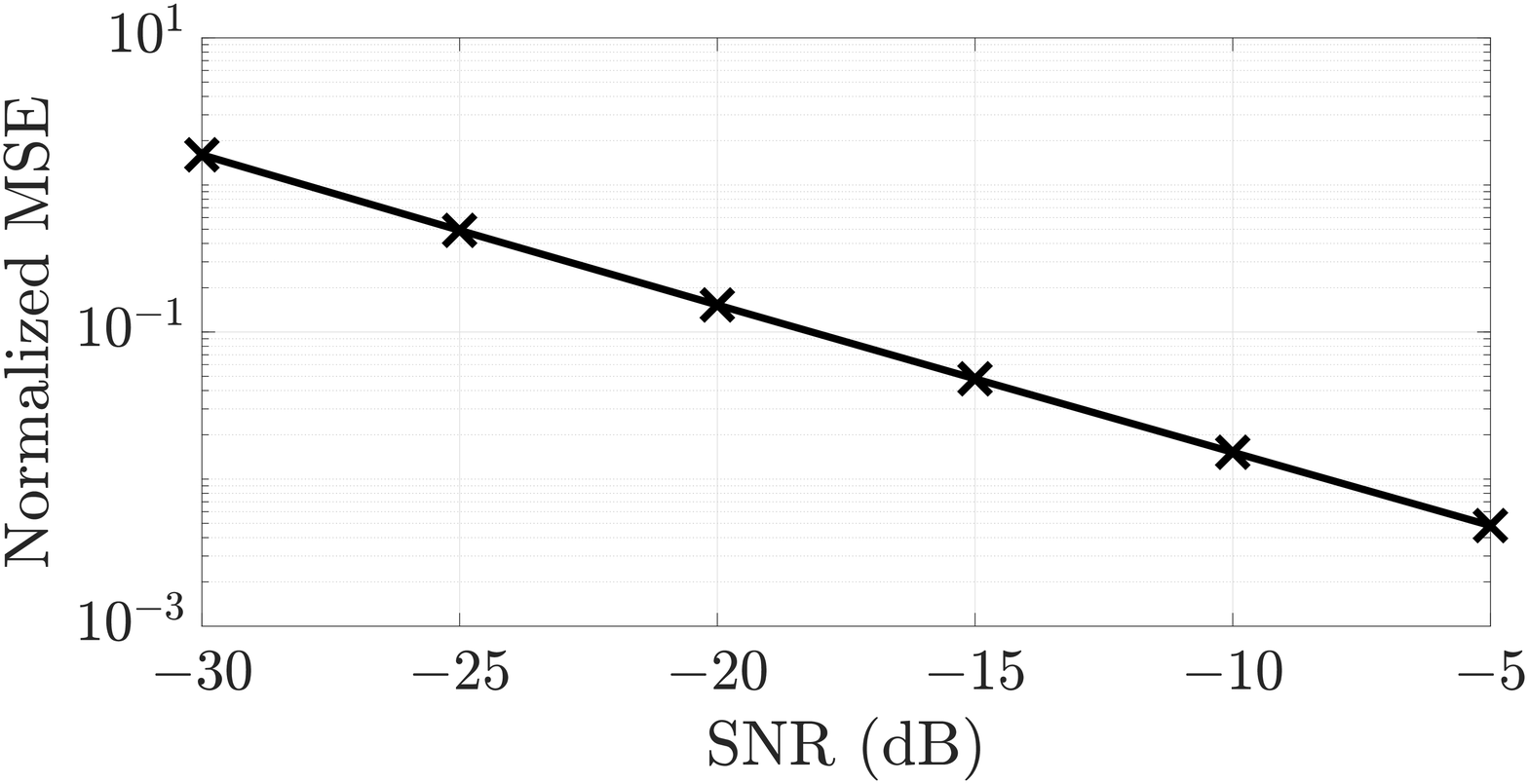}
\captionsetup{justification=centering}
\caption{\small $32 {\times} 32$ channel with $L{=}3$}
\label{fig:MSE_32x32_9L_3P}
\end{subfigure}%
\captionsetup{justification=centering}
\caption{\small Normalized mean squared error (MSE).}
\end{figure*}

In this section we provide extensive simulation results to evaluate the performance of our devised Beam Discovery approach.
We begin this section by introducing the simulation setup and system parameters before defining the performance metrics we use for evaluation.

\subsection{Simulation Setup and Parameters}
We consider an $n_r {\times} n_t$ mm-wave channel with noise power per symbol $N_0 {=} {-}95$dBm, and average path loss $\mu {=} 136dB$.
A maximum number of $L$ (strong) paths exist between TX and RX.
A strong path is defined such that its path attenuation is not higher than $14$dB above the average path loss, i.e. total path loss is at most $150$dB.


Let $\tau$ be the time duration of a pilot sequence of one measurement. For simplicity, let $\tau {=} 1$.
Also, recall that \textit{SNR} is defined for a single path (see equation (\ref{eqn:SNRl})), and that $P$ is the corresponding transmitted power (i.e., per path). Let the total transmitted power be $P_t$, where $P_t$ is an integer multiple of $P$ that depends on the number of combined transmit and receive directions (recall Fig. \ref{fig:RxComb}).
Then, the total energy required for beam discovery is $E {=} m P_t \tau {=} mP_t$, where $m$ is the total number of measurements\footnote{This formula for total energy assumes equal $P_t$ for all measurements. Depending on the employed LBC, this might not always be the case. More generally, we can find the total energy to be: $E = \sum_i m_i P_{t_i}$, where $m_i$ is the number of measurements with total transmit power $P_{t_i}$.}.
Let the normalized energy be $E_t \triangleq \frac{E}{N_0} |\alpha_{\min}/\mu|^2 $.


To map the channel measurements to their corresponding angular channels, we use the search method presented in Section \ref{subsub:SearchMethod}.
Finally, for every simulation scenario, we obtain the average performance across $10^5$ runs.

\subsection{Performance metrics}
To asses the performance of the proposed beam discovery method, we mainly focus on three basic criteria, namely, \textit{\textbf{accuracy of beam discovery}}, \textit{\textbf{number of measurements}}, and \textit{\textbf{accuracy of path gain value estimates}}.
To that end, we use the following performance metrics:
\begin{enumerate}[i)]
\item \textbf{\textit{Number of measurements:}} This represents the number of pilots sent from TX to discover the paths to RX.
\item \textbf{\textit{Probability of strongest $k$ beams discovery:}} This denotes the probability of correctly identifying the directions of $k$ strong reflectors among $L$. There are two cases we consider pertaining to the possibility of the algorithm identifying exactly $k$ directions or more than $k$ directions at the output:\\
\textbf{1) perfect $k$ beam discovery:} where exactly $k$ true paths are discovered with no incorrect paths among them.\\
\textbf{2) all $k$ beam discovery:} where $k$ true paths are discovered with potentially more incorrectly identified paths.
\item \textit{\textbf{Number of incorrect beams:}} Due to the possibility of obtaining a combination of correct and incorrect paths, it is important that we have as few incorrect beams as possible since further refinement would be made easier.
\item \textbf{\textit{Normalized mean squared error (MSE):}} $\frac{\left\Vert \boldsymbol{Q^a} - \widehat{\boldsymbol{Q}}^{\boldsymbol{a}}\right\Vert^2_F}{\left\Vert \boldsymbol{Q^a}\right\Vert^2_F}$. Measurement errors occur in the form of 1) imperfect estimates of path gains and phases, and 2) incorrect beam discovery. Hence, MSE provides an inclusive metric for how close the estimated channel matrix is to the true one.
\end{enumerate}

Measurement errors mainly occur due to two contributing factors. The first is \textit{\textbf{measurement noise}}, and the second is \textit{\textit{\textbf{quantization}}} (recall that we assume the measurements to be quantized using mid-tread ADC quantizers with $2^b{+}1$ levels).\\
In Sections \ref{SinglePath_min} and \ref{MultiPath_min}, we assess the performance of Beam Discovery approach against \textbf{\textit{only}} the effect of measurements noise. We do so by assuming a perfect, infinite resolution ADC.
Then, in Section \ref{sec:quantizationEffect}, we investigate the system performance at different ADC resolution levels.
This separate investigation of sources of errors allows us to understand how each source affects the performance. Thus, enabling full realization of potential gains of Beam Discovery approach.

\subsection{Single-path channels} \label{SinglePath_min}
Consider a $15 {\times} 15$ mm-wave channel with $L {=} 1$ path between TX and RX. Hence, the parity check matrix of $(15,11,3)$ Hamming code can be used to  design both the precoders, $\boldsymbol{f_j}$, and rx-combiners, $\boldsymbol{w_i}$, i.e., $\boldsymbol{H_1}$ and $\boldsymbol{H_2}$, are identical.
Hence, we need a number of TX measurements $m_1$, which is identical to the number of RX measurements $m_2 {=} 15{-}11 {=} 4$. Hence, the total number of measurements is $m {=} 16$. On the other hand, the exhaustive scanning method requires $225$ measurements to inspect every possible TX-RX beam combination. Thus, our approach results in ${\approx} 92.8 \%$ reduction in the required number of measurements.


At different \textit{SNR} values, we plot the \textbf{\textit{probability of error}} curves for: i) \textit{Perfect} beam discovery where only the single strongest path ($k{=}1$) is correctly identified, and ii) \textit{all} beam discovery where the strongest path is correctly identified among potentially other misidentified paths. Fig. \ref{fig:beamDetProb_15x15_9L_1P} shows those curves.
We observe that both curves are on top of each other which indicates that the strongest path is either correctly detected or is completely missed.
Moreover, at all \textit{SNR} values ${\geq} -5$dB, the probability of error is lower than $10^{-5}$, and hence, is not shown here since the shown figures are the averages of $10^5$ simulation runs.

In Fig. \ref{fig:MSE_15x15_9L_1P}, we plot the normalized mean squared error of the channel estimate $\widehat{\boldsymbol{Q}}^{\boldsymbol{a}}$.
The very high values at low signal to noise ratios indicate that $\widehat{\boldsymbol{Q}}^{\boldsymbol{a}}$ has large components at truly zero components in $\boldsymbol{Q^a}$ and/or large components in $\boldsymbol{Q^a}$ are not represented in $\widehat{\boldsymbol{Q}}^{\boldsymbol{a}}$.
Nevertheless, MSE drops steadily fast as \textit{SNR} increases; indicating improved channel estimation.

When we talk about the possibility of misidentified beams for the all beam discovery metric, it is crucial to have a small number of incorrect beam which would facilitate further refinement.
Interestingly, for this scenario, since the error performance of perfect and all beam discovery are the same, we do not have any misidentified paths besides the correct one. Nevertheless, this is not always the case as we will see in further investigated scenarios.

\subsection{Multi-path  Channels}
\label{MultiPath_min}

First, consider an $8 {\times} 8$ channel with $L {=} 2$ paths. For this scenario, we use an $(8,2,5)$ code for both $\boldsymbol{H_1}$ and $\boldsymbol{H_2}$. With this code, a total number, $36$, of measurements is needed for beam discovery. Compared with the $64$ measurements needed for exhaustive scanning, we achieve ${\approx} 43.7 \%$ reduction in the number of measurements under this scenario.

Since we investigate a channel that potentially has two strong paths, we evaluate the probability of error of picking one correct strong path ($k{=}1$) as well as picking two strong paths ($k{=}2$).
Fig. \ref{fig:beamDetProb_8x8_9L_2P} depicts the corresponding probability of error of the perfect and all $k$ beam discovery metrics.
Unlike single-path channels, there exists a wider gap between perfect and all beam discovery curves for the $k{=}1$ metric; which indicates higher vulnerability to picking incorrect paths. On the other hand, for $k{=}2$, the error performance of the perfect and all beam discovery metrics are almost on top of each other.
In Fig. \ref{fig:MSE_8x8_9L_2P}, similar trend for normalized MSE is obtained where MSE steadily drops as \textit{SNR} increases.

Recall that in the $15 {\times} 15$ single-path channel investigation, no incorrect paths were obtained alongside correctly identified strong paths.
This behavior is not replicated for the $8{\times}8$ channel under investigation.
For instance, at ${-}10$dB we obtain a maximum of $2$ misidentified paths.
Further, the probability of obtaining incorrect paths at ${-10}$dB is $\approx 0.04637$.

We further consider a larger array with dimensions $32{\times}32$ and $L{=}3$ paths. We use a $(32,16,8)$ Reed-Muller code to design both $\boldsymbol{H_1}$ and $\boldsymbol{H_2}$. This corresponds to $m_1 {=} m_2 {=} 16$, i.e., total number of measurements $m {=} 265$. This is $75\%$ fewer measurements needed compared to exhaustive scanning which requires $1024$ measurements for beam discovery.

The probability of error for perfect and all $k{=}1,2,3$ beam discovery are shown in Fig. \ref{fig:beamDetProb_32x32_9L_3P}.
We notice a faster decay rate for the probability of error. This behavior is due to the higher gain of the TX and RX antenna arrays; which increases the receive signal to noise ratios compared to small arrays.
The normalized MSE is shown in \ref{fig:MSE_32x32_9L_3P} have similar trend to the previous investigated scenarios.

\subsection{Effect of Quantization} \label{sec:quantizationEffect}

\begin{figure}[t]
\centering
\captionsetup{justification=centering}
\includegraphics[width=0.8\linewidth]{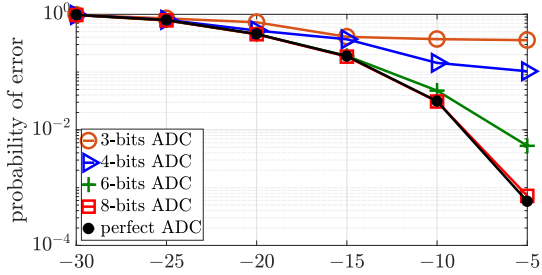}
\caption{\small Perfect $k{=}2$ beam discovery at different quantization resolution ($8 {\times} 8$ channel with $L{=}2$)}
\label{fig:quantization_efect_8x8_2P}
\end{figure}

In this section, both sources of errors are incorporated.
Specifically, we analyze the system performance at different ADC resolution levels.
We will show that very low resolution ADCs can have detrimental effect on performance.
Thus, a natural question that we try to answer in this study is: \textbf{\textit{How far should we increase the resolution of quantizers in order to unlock the full potential of the Beam Discovery approach?}}

Recall that we use mid-tread ADCs with $2^b{+}1$ quantization levels ($b$ stands for the number of bits required to represent the ADC output (approximately)).
We limit our discussion to the case of $8{\times}8$ channels with $L{=}2$ paths since its results are representative of the other previously investigated scenarios.
For clarity and legibility of figures, we only plot the perfect $k{=}2$ beam discovery for $b = 3,4,6,8$ bits i.e., the corresponding number of quanization levels is $9,17,65,257$, respectively. We also plot the corresponding probability of error using a perfect ADC (i.e., $b \rightarrow \infty$). These curves are shown in Fig. \ref{fig:quantization_efect_8x8_2P}.

We find that, at $b{=}3$, the probability of error is very high and does not improve with increasing \textit{SNR}. Hence, quantization is the dominant source of errors.
Then, as the resolution of ADCs increase, significant performance improvement can be achieved.
For instance, while $b{=}4$ still do not produce very good probability of error (with increasing \textit{SNR}), a huge leap in performance can be obtained using ADCs with only $b{=}6$ bits.
Moreover, at $b{=}8$, we approach the performance of perfect ADCs.
Note that we just need $2$ ADCs as per our proposed receiver architecture (see Fig. \ref{fig:architecture}).

\subsection{Error Correction}

\begin{figure}[t]
\centering
\includegraphics[width=0.8\linewidth]{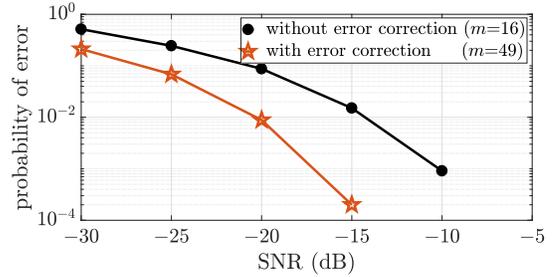}
\caption{\small Beam detection probability ($15 {\times} 15$ channel with $L{=}1$)}
\label{fig:beamDProb_ErrCorr_15x15_9L_1P}
\end{figure}
\begin{figure}[t]
\centering
\includegraphics[width=0.75\linewidth]{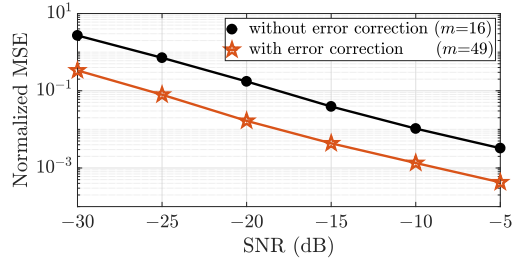}
\caption{\small MSE ($15 {\times} 15$ channel with $L{=}1$)}
\label{fig:MSE_ErrCorr_15x15_9L_1P} 
\end{figure}

\begin{figure}[t]
\centering
\includegraphics[width=0.8\linewidth]{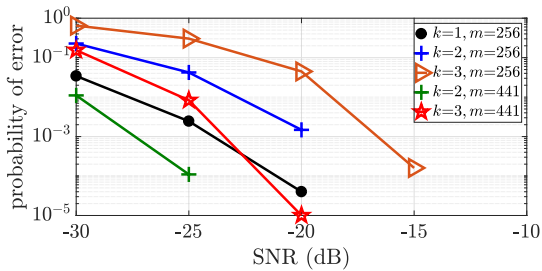}
\caption{\small Beam detection probability ($32 {\times} 32$ channel with $L{=}3$)}
\label{fig:beamDProb_ErrCorr_32x32_9L_3P}
\end{figure}
\begin{figure}[t]
\centering
\includegraphics[width=0.75\linewidth]{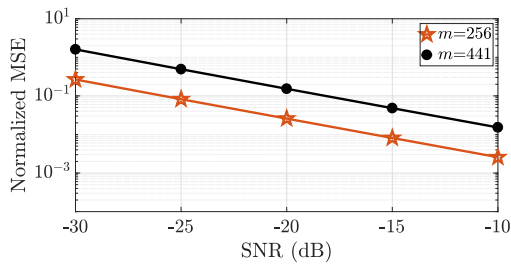}
\caption{\small MSE ($32 {\times} 32$ channel with $L{=}3$)}
\label{fig:MSE_ErrCorr_32x32_9L_3P} 
\end{figure}

In this section, we investigate the performance of Beam Discovery with the \textbf{\textit{error correction}} technique proposed in Section\ref{ErrorCorr}. Recall that error correction is a \textit{channel-coding-like} technique that allows for improving the error performance on the expense of increased number of measurements.
Two different scenarios are investigated; namely, $15{\times}15$ single-path channels, and $32{\times}32$ triple-path channels.

\subsubsection{Single-Path Channel}
Consider the $15 {\times} 15 $ single-path channel we studied in Section \ref{SinglePath_min}.
Recall that we used the parity check matrix of $(15,11,3)$ Hamming code for both $\boldsymbol{H_1}$ and $\boldsymbol{H_2}$ which resulted in syndromes $\boldsymbol{y_s}$ of length $m_1 {=} m_2{=} 4$.
Now, we need to encode sequences of length $4$ into longer sequences $\boldsymbol{y_s^{\nu}}$ using a systematic code. Conveniently, we can use the $(7,4,3)$ Hamming code which maps sequences of length $4$ into sequences of length $7$.
The corresponding $\boldsymbol{H^{\nu}_1}$ and $\boldsymbol{H^{\nu}_2}$ matrices are of size $7 {\times} 15$ and we have that $m_{c_1} {=} m_{c_2} {=} 7$. Hence we have a total number of measurements for Beam Discovery with error correction $m_{c} {=} 49$. This is ${\approx} 78.2\%$ fewer measurements compared to exhaustive scanning. Recall that the number of measurements without error correction is $16$.

The probability of error for perfect $k{=}1$ beam discovery is depicted in Fig. \ref{fig:beamDProb_ErrCorr_15x15_9L_1P}. A notable performance improvement over the $m{=}16$ case is obtained. That is, at the same \textit{SNR}, significantly lower probability of error is achieved. This performance improvement is also reflected in the MSE curves in Fig. \ref{fig:MSE_ErrCorr_15x15_9L_1P}.

\subsubsection{Multi-Path Channel}
For the multi-path scenario, we study the $32 {\times} 32$ channel with $L {= 3}$ paths.
Recall that, in Section\ref{MultiPath_min}, we use a $(32,16,8)$ Reed-Muller code for which the parity check matrices $\boldsymbol{H_1} {=} \boldsymbol{H_2}$ are of size $16 {\times} 32$. Under this setting we obtain $75\%$ reduction in the number of channel measurements compared to exhaustive scanning ($256$ instead of $1024$ measurements).
To add the error correction capability, we encode the channel syndromes using a $(21,16,3)$ code (a subcode of the $(31,26,3)$ Hamming code). We obtain $\boldsymbol{H^{\nu}_1} {=} \boldsymbol{H^{\nu}_2}$ of size $21 {\times} 32$. Thus, $m_{c_1} {=} m_{c_2} {=} 21$, i.e., $m_c{=}441$, which gives a reduction of $\approx 57\%$ in number of measurements compared to exhaustive scanning.

For clarity, we only plot the probability of error for perfect $k{=}1,2,3$ beam discovery shown in Fig. \ref{fig:beamDProb_ErrCorr_32x32_9L_3P}. 
Note that at $m_c{=}441$, the $k{=}1$ perfect beam discovery achieves error probability below $10^{-5}$, hence, it is not shown in Fig. \ref{fig:beamDProb_ErrCorr_32x32_9L_3P}.
We notice a huge performance improvement over the $m{=}265$ case, that is, at fixed \textit{SNR} we obtain at least an order of magnitude improvement in the probability of error.
We also see a corresponding improvement in MSE depicted in Fig. \ref{fig:MSE_ErrCorr_32x32_9L_3P}.


\subsection{Comparison to Exhaustive Scanning}

\begin{figure}[t]
\centering
\includegraphics[width=0.75\linewidth]{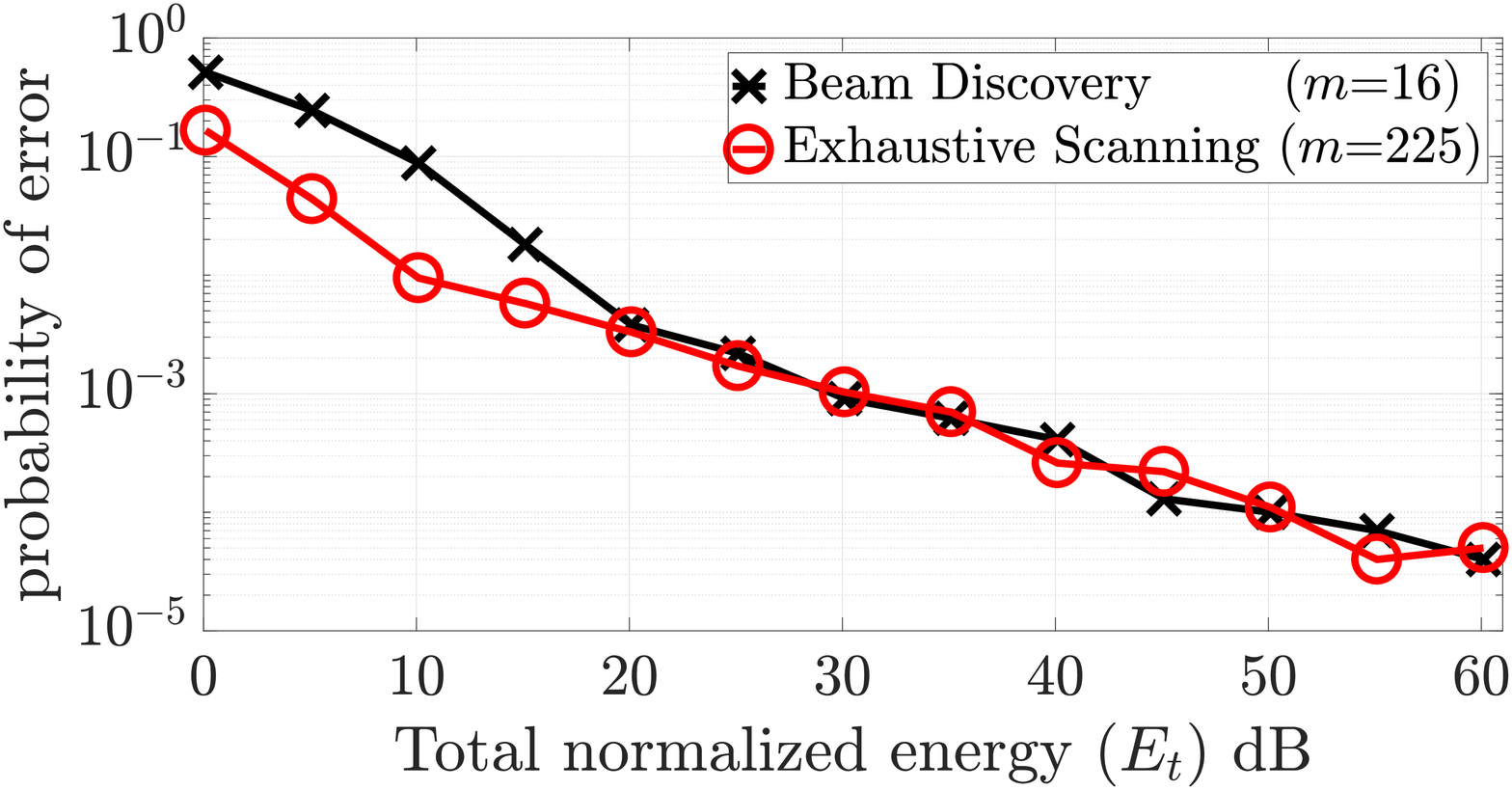}
\captionsetup{justification=centering}
\caption{\small Perfect $k{=}1$ beam discovery.\\
Beam Discovery vs. Scanning ($15 {\times} 15$ channel with $L{=}1$)}
\label{fig:beamDProb_m=16_vs_225_InfL_1P} 
\end{figure}

\begin{figure}[t]
\centering
\includegraphics[width=0.75\linewidth]{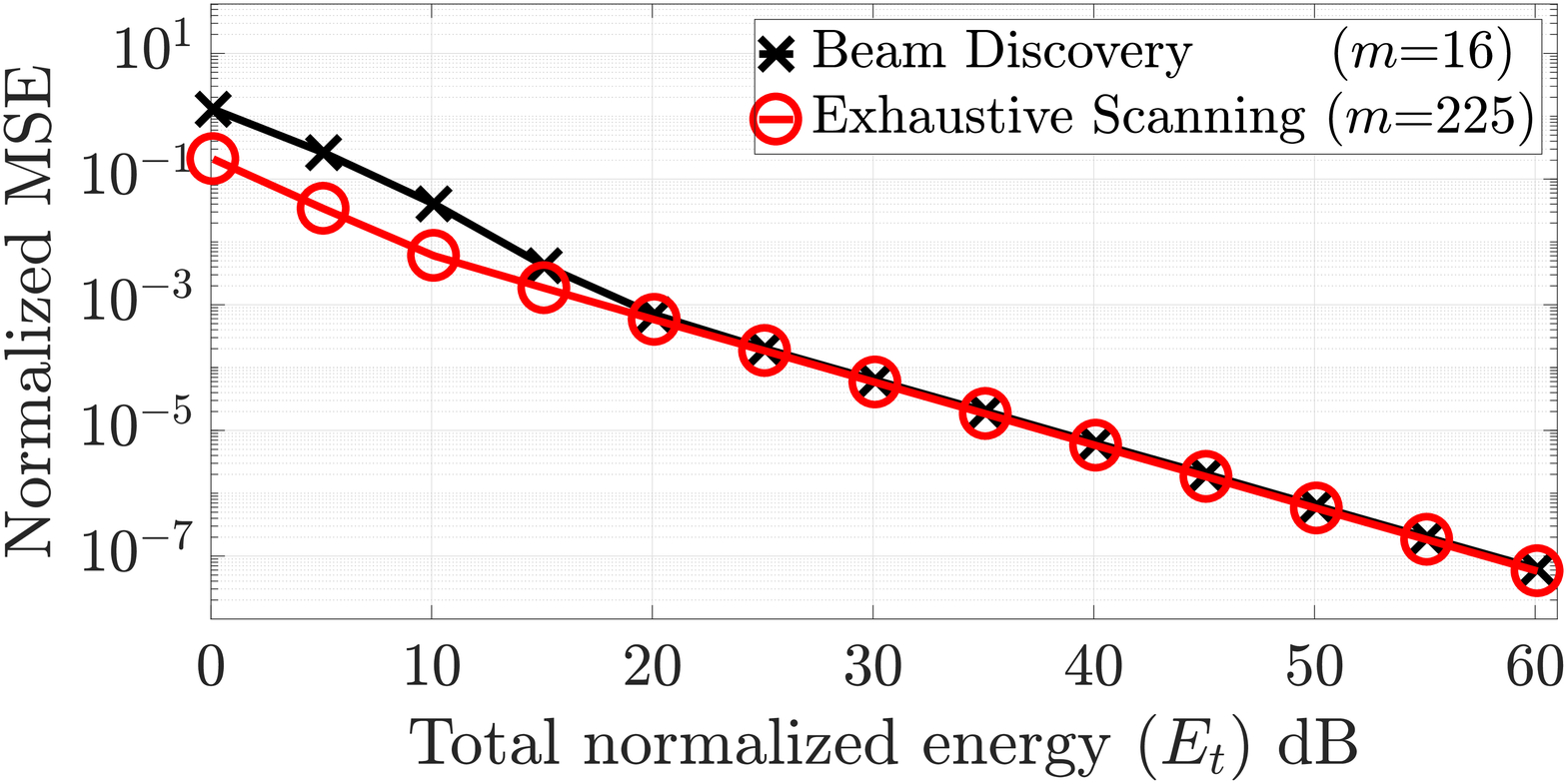}
\captionsetup{justification=centering}
\caption{\small Normalized MSE.
\\Beam Discovery vs. Scanning ($15 {\times} 15$ channel with $L{=}1$)}
\label{fig:MSE_m=16_vs_225_InfL_1P} 
\end{figure}


In the previous analysis, we have shown that our devised approach requires significantly fewer measurements for beam discovery when compared to Exhaustive Scanning.
Although Scanning entails a large number of measurements, its transmit power $P_t$ is kept low since it only requires pilot transmission over one TX-RX beam combination. On the other hand, our Beam Discovery approach requires transmission over multiple TX-RX beam combinations, which necessitates more invested transmission power in order to keep \textit{SNR} equal to the one obtained from Scanning\footnote{In other words, if we keep $P_t$ fixed for both schemes, then Beam Discovery will operate at lower \textit{SNR}. This is common in mm-wave channel estimation.}. 

In this section, we evaluate the performance of our proposed approach vs.  $E_t$, and compare it against Exhaustive Scanning.
Recall that $E_t$ is the total normalized energy required for the beam discovery process.
To neutralize the effect of quantization, we assume that perfect ADCs are used for both schemes.
We limit our discussion to only the $15{\times}15$ single-path channel.
%
%

In Fig. \ref{fig:beamDProb_m=16_vs_225_InfL_1P}, we plot the probability of error for perfect $k{=}1$ beam discovery with $m=16$, and for Exhaustive scanning ($m=225$).
We find that at $E_t$ above $20$dB, we achieve almost the same error performance as Scanning, yet, with $92.8\%$ fewer measurements. This is further emphasized by MSE curves shown in Fig. \ref{fig:MSE_m=16_vs_225_InfL_1P}.

%

\section{Conclusion}
\label{DiscConc}
This work provides a solution for the mm-wave channel estimation problem by exploiting its sparse nature in the angular domain. The proposed solution is a beam discovery technique that is similar to error discovery in channel coding. We show that our proposed technique can significantly reduce the number of measurements required for reliable channel estimation.
Our solution takes into account the size of TX/RX arrays and the sparsity level of the channel. We determine the number of measurements and the design of each measurement in a deterministic way; based on parity check matrices of appropriately selected LBCs.
Under no measurement errors, our solution is guaranteed to find all available beams (paths) between TX and RX. However, due to the presence of channel noise and quantization (ADCs), measurement errors occur, which might cause incorrect beam discovery. Hence, we assess the performance of the proposed scheme under different levels of \textit{SNR} and ADC resolutions.
We further provide a technique for error correction that is also inspired by channel coding.\\
A special case of uncoded discovery within our general coded discovery framework is Exhaustive Scanning.
We compare our solution against Scanning and find that we approach its error performance under the same total energy expenditure.

\appendices
\section{Proof of Lemma \ref{lemma:Uni_ErrCrr}} 
\label{append:proof_LemmaUniErrCrr}
\begin{proof}
We start with the following \textit{claim:} All error syndromes $\boldsymbol{s}^{\boldsymbol{\nu}}$ corresponding to correctable error patterns $\boldsymbol{e} \in \mathcal{E}_C$ are unique. Recall that $\mathcal{E}_C$ is the set of correctable error patterns of code $C$
(Eq. (\ref{eqn:setOfCorrectableErrors})).
Suppose that the claim is true.
Then, assume towards contradiction that $\exists \: \boldsymbol{s}_t^{\boldsymbol{\nu}}$ given by a linear combination of $e_n$ linearly \textbf{dependent} vectors $\boldsymbol{s^{\nu}}^{(i)}$. Hence, $\exists$ another vector $\boldsymbol{s}_t^{*\boldsymbol{\nu}}$ composed of linear combination of only linearly independent vectors such that $\boldsymbol{s}_t^{\boldsymbol{\nu}} = \boldsymbol{s}_t^{*\boldsymbol{\nu}}$.\\
So, $\boldsymbol{s}_t^{*\boldsymbol{\nu}}$ corresponds to an error pattern $\boldsymbol{e}_t^*$ with number of errors $< e_n \Rightarrow \boldsymbol{e}_t^* \in \mathcal{E}_C$.\\
Since $\boldsymbol{e}_t \neq \boldsymbol{e}_t^*$ and both $\boldsymbol{e}_t, \boldsymbol{e}_t^* \in \mathcal{E}_C$, then by our claim
$\Rightarrow \boldsymbol{s}_t^{\boldsymbol{\nu}} \neq \boldsymbol{s}_t^{*\boldsymbol{\nu}}$. Hence, we arrive at  a contradiction.
\end{proof}
\vspace{-5pt}
It remains to show that our claim is true.
\vspace{-5pt}
\begin{proof}[Proof of claim]
Recall that both codes $C$ and $C_c$ are binary linear block codes. The codewords of $C$ span a linear subspace with dimension $k$ of the vector space $\{GF(2)\}^{n}$ and the codewords of $C_c$ span a linear subspace with dimension $m$ of the vector space $\{GF(2)\}^{m_c}$.
Hence, all operations performed using parity check and generator matrices, codewords, error sequences and error syndromes are all linear over $GF(2)$.

By inspection of Eq. (\ref{eqn:encodedSyndrome}), we have that:
\begin{itemize}
\item Equality $(a)$ follows by our definition of encoding the error syndromes $\boldsymbol{s}$ into encoded syndromes $\boldsymbol{s^{\nu}}$.
Since $\boldsymbol{G_c}$ is $m{\times} m_c$ with linearly independent rows \cite{van2012introduction}, we have that if $\boldsymbol{s}_1 \neq \boldsymbol{s}_2$, then $\boldsymbol{s_1^{\nu}} \neq \boldsymbol{s_2^{\nu}}$ with $\boldsymbol{s_i^{\nu}}$ being the corresponding encoded $\boldsymbol{s_i}$ sequence.

\item Equality $(b)$ follows by substituting for $\boldsymbol{s}$ with its equivalent linear operation $\boldsymbol{eH}^T$ (recall Footnote \ref{LBC}).
We also have that $\forall \boldsymbol{e} \in \mathcal{E}_C  \; \exists! \boldsymbol{s}: \boldsymbol{s} \in \{GF(2)\}^{m}$.

\item Equality $(c)$ follows by linearity of $\boldsymbol{H}$ and $\boldsymbol{G_c}$.
\end{itemize} 

Hence, $\forall \boldsymbol{e_1}, \boldsymbol{e_2} {\in} \mathcal{E}_C$,  $\boldsymbol{e_1} {\neq} \boldsymbol{e_2} \Leftrightarrow \boldsymbol{s_1} {\neq} \boldsymbol{s_2} \Leftrightarrow \boldsymbol{s_1^{\nu}} {\neq} \boldsymbol{s_s^{\nu}}$.
\end{proof}

\section{Proof of Proposition \ref{prop:Distance_ErrCrr}}
\label{append:proof_Prop}
\begin{proof}
We need to show that
\begin{align}
\norm{\boldsymbol{y_{s_1}^{\nu}} - \boldsymbol{y_{s_2}^{\nu}}}  &\geq
\norm{\boldsymbol{y_{s_1}}       - \boldsymbol{y_{s_2}}      }   \\
\Longleftrightarrow \norm{\boldsymbol{H^{\nu}} \boldsymbol{q^a_1} - \boldsymbol{H^{\nu}} \boldsymbol{q^a_2}}  &\geq
\norm{\boldsymbol{H} \boldsymbol{q^a_1} - \boldsymbol{H} \boldsymbol{q^a_2}} 
\\
\Longleftrightarrow \norm{\boldsymbol{H^{\nu}} \left( \boldsymbol{q^a_1} - \boldsymbol{q^a_2} \right) }^2  &\geq \norm{\boldsymbol{H} \left( \boldsymbol{q^a_1} - \boldsymbol{q^a_2} \right)}^2 \label{eqn:normSquare}
\end{align}
Let $\boldsymbol{v} = \boldsymbol{q^a_1} - \boldsymbol{q^a_2} $,
then, Eq. (\ref{eqn:normSquare}) is true if and only if
\begin{align}
\left( \boldsymbol{H^{\nu}} \boldsymbol{v} \right)^T \left( \boldsymbol{H^{\nu}} \boldsymbol{v} \right)  &\geq \left( \boldsymbol{H^{\nu}} \boldsymbol{v} \right)^T \left( \boldsymbol{H} \boldsymbol{v} \right)  \\
\Longleftrightarrow \boldsymbol{v}^T  \left( \boldsymbol{H^{\nu}} \right)^T \boldsymbol{H^{\nu}} \boldsymbol{v}  &\geq \boldsymbol{v}^T  \boldsymbol{H}^T \boldsymbol{H} \boldsymbol{v}
\end{align}
\vspace{-20pt}
\begin{align}
\Longleftrightarrow \boldsymbol{v}^T  \left( \left( \boldsymbol{H^{\nu}} \right)^T \boldsymbol{H^{\nu}} - \boldsymbol{H}^T \boldsymbol{H} \right) \boldsymbol{v}  &\geq 0 \label{eqn:prop4_subs} \\
\Longleftrightarrow \quad \left( \boldsymbol{H^{\nu}} \right)^T \boldsymbol{H^{\nu}} - \boldsymbol{H}^T \boldsymbol{H}   &\succeq 0, \label{eqn:prop4_subs}
\end{align}
i.e., $\left( \boldsymbol{H^{\nu}} \right)^T \boldsymbol{H^{\nu}} {-} \boldsymbol{H}^T \boldsymbol{H}$ is positive semi-definite.
Suppose $\boldsymbol{G_c}$ is of size $m {\times} m_c$.
Since $\boldsymbol{G_c}$ is in standard form\footnote{For any LBC code $C1$, we can find an equivalent \textit{systematic} code $C2$ using row reduction and column reordering operations \cite{van2012introduction}. A systematic code is an LBC with generator matrix structure given by Eq. (\ref{eqn:G_stdForm}).} (by assumption), then it can be written in block matrix representation as
\begin{equation}
\label{eqn:G_stdForm}
\boldsymbol{G_c} =
\begin{pmatrix}
\boldsymbol{I} & \vline & \boldsymbol{P}
\end{pmatrix},
\end{equation}
where $\boldsymbol{I}$ is the $m {\times} m$ identity matrix and $\boldsymbol{P}$ is of size $m {\times} m_c{-}m$.
Then, we can find $\boldsymbol{H^{\nu}}$ to be given as
\begin{align}
\boldsymbol{H^{\nu}} &= \boldsymbol{G_c}^T \boldsymbol{H} \pmod {2} \; \;
= 
\begin{pmatrix}
\boldsymbol{I}  \\ \boldsymbol{P}^T
\end{pmatrix}
\boldsymbol{H} \pmod {2} \\
&= 
\begin{pmatrix}
\boldsymbol{H}  \\ \boldsymbol{P}^T \boldsymbol{H}
\end{pmatrix}
 \pmod {2}
= 
\begin{pmatrix}
\boldsymbol{H}  \\ \boldsymbol{P}_m
\end{pmatrix}
\end{align}
where $\boldsymbol{P}_m = \boldsymbol{P}^T \boldsymbol{H} \pmod {2}$.
Hence, we have that
\begin{equation}
\label{eqn:_blockMatrix}
\left( \boldsymbol{H^{\nu}} \right)^T \boldsymbol{H^{\nu}} = \boldsymbol{H}^T \boldsymbol{H} + \boldsymbol{P}^T_m \boldsymbol{P}_m
\end{equation}
Hence, we get that $\forall \boldsymbol{v} \in \mathbb{C}^n$,
\begin{align}
\boldsymbol{v}^T  \left( \left( \boldsymbol{H^{\nu}} \right)^T \boldsymbol{H^{\nu}} - \boldsymbol{H}^T \boldsymbol{H} \right)\boldsymbol{v}
&=
\boldsymbol{v}^T  \boldsymbol{P}^T_m \boldsymbol{P}_m \boldsymbol{v} \\
&= \norm{\boldsymbol{P}_m \boldsymbol{v}}^2 \geq 0
\end{align}
Hence, Eq. (\ref{eqn:prop4_subs}) is satisfied which completes the proof.
\end{proof}





\end{document}